\documentclass[11pt]{article}%
\usepackage{bbm,url,microtype,amssymb,amsthm,mathtools,mathrsfs,graphicx,float,color}
\usepackage{fullpage}
\usepackage{lmodern}
\usepackage[T1]{fontenc}
\usepackage[USenglish]{babel}
\usepackage{authblk}
\usepackage{booktabs}
\usepackage[hidelinks]{hyperref}
\usepackage{cleveref}
\usepackage{bm}
\usepackage{placeins}

\providecommand{\U}[1]{\protect\rule{.1in}{.1in}}
\newtheorem{thm}{Theorem}\crefname{thm}{Theorem}{Theorems}
\newtheorem{lem}[thm]{Lemma}\crefname{lem}{Lemma}{Lemmas}
\crefname{prp}{Proposition}{Propositions}
\crefname{cor}{Corollary}{Corollaries}
\crefname{prb}{Problem}{Problems}
\newtheorem{dfn}[thm]{Definition}\crefname{dfn}{Definition}{Definitions}
\crefname{obs}{Observation}{Observations}
\crefname{section}{Section}{Sections}
\crefname{appendix}{Appendix}{Appendices}
\numberwithin{equation}{section}
\let\oldref\ref
\renewcommand{\ref}[1]{(\oldref{#1})}
\DeclareMathOperator{\tr}{tr}

\newcommand{\R}{\mathbb{R}}
\newcommand{\N}{\mathbb{N}}
\newcommand{\C}{\mathbb{C}}

\newcommand{\ket}[1]{\vert #1\rangle }
\newcommand{\bra}[1]{\langle #1 \vert}
\newcommand{\ketbra}[1]{\ket{#1} \bra{#1}}
\newcommand{\abs}[1]{{\left\vert{#1}\right\vert}}
\newcommand{\norm}[1]{{\left\Vert{#1}\right\Vert}}
\newcommand{\ot}{\otimes}
\newcommand{\id}{\mathbbm 1}
\newcommand{\Nin}{{n_{\mathrm{in}}}}
\newcommand{\Nout}{{n_{\mathrm{out}}}}
\newcommand{\Nparty}{{n_p}}

\usepackage{tikz} %
\usetikzlibrary{positioning,calc,shapes.arrows}

\begin{document}
\title{Efficiently Optimizing the Quantum Value of Bell Inequalities \\ using Batched Gradient Descent}
\author[1,2]{Xinyu Xu\thanks{These authors contributed equally.}}
\author[3]{Ping Zhu\raisebox{0.45ex}{\scriptsize *}}
\author[4]{Weikang Li}
\author[5,6]{Pierre Pocreau}
\author[3]{Dalu Ding\thanks{jerry.ding@jmuse.cn}}
\author[7,1]{Dawei Ding\thanks{daweiding@fudan.edu.cn}}
\affil[1]{Shanghai Institute for Mathematics and Interdisciplinary Sciences (SIMIS), \newline Shanghai 200433, China}
\affil[2]{Research Institute of Intelligent Complex Systems, Fudan University, \newline Shanghai 200433, China}
\affil[3]{Jmuse Technologies, Shanghai 200000, China}
\affil[4]{Center for Quantum Information, IIIS, Tsinghua University, Beijing 100084, China}
\affil[5]{Inria, Université Grenoble Alpes, 38000 Grenoble, France}
\affil[6]{CNRS, Grenoble INP, LIG, Université Grenoble Alpes, 38000 Grenoble, France}
\affil[7]{Center for Mathematics and Interdisciplinary Sciences, Fudan University, \newline Shanghai 200433, China}

\date{}
\maketitle

\begin{abstract}
    Exploring the rich and complex world of Bell inequalities requires numerical methods that scale effectively with the number of parties, inputs, and outputs.
    Such methods can be used to study large-scale Bell inequalities beyond the reach of analytic methods,
    to design experiments for violating more complicated Bell inequalities, and to enable applications of Bell inequality violation. Such applications include randomness generation and multi-agent coordination with restricted communication. However, current general purpose optimizers such as the see-saw method struggle to handle Bell inequalities with only a dozen inputs and outputs. In this paper we introduce an optimizer for the quantum value of a Bell inequality based on batched gradient descent (BGD). More precisely, our BGD optimizer is a differentiable search over feasible quantum strategies in which states and projective measurements are generated from unconstrained variables. The Bell expression is evaluated by a direct tensor contraction without forming the prohibitively large Bell operator. This formulation makes parallel random restarts inexpensive and is thus naturally implemented by a GPU. We evaluate our optimizer against the see-saw method and find a significant speedup for a wide variety of Bell inequality families. Implemented on a GPU, our BGD optimizer can optimize Bell inequalities in our numerical experiments with more than a thousand inputs or outputs within minutes.
    For these large-scale Bell inequalities, we also efficiently compute the classical value by enumerating possible strategies using a GPU or expressing the optimization as a mixed-integer linear program, for which we can use off-the-shelf solvers such as Gurobi.
    Our optimizers can serve as numerical evaluation tools for Bell experiments and 
    multi-agent coordination problems modeled by Bell inequalities with many inputs and outputs, a common feature of real-world scenarios such as high frequency trading and distributed systems.
\end{abstract}

\clearpage
\tableofcontents

\section{Introduction}
Bell inequalities~\cite{bell1964einstein} are constraints on the correlations between multiple parties obeying a local hidden variable theory. A series of experiments using quantum systems to violate these inequalities
definitively disproved any classical description of nature~\cite{clauser1969proposed,aspect1982experimental,weihs1998violation}. In fact, the significance of this conclusion motivated an entire field of study dedicated to Bell nonlocality~\cite{brunner2014bell} and played an important role in the development of quantum information science~\cite{aspect2022nobel}. Bell inequality violation is also the basis for certain quantum technologies, including device-independent quantum key distribution~\cite{zapatero2023advances},
certified randomness~\cite{AcinMasanes2016},
constant-depth quantum computation~\cite{bravyi2018quantum}, and multi-agent coordination with restricted communication~\cite{ding2026quantum}.

A systematic study of Bell inequalities is a challenging task~\cite{Rosset2014Classifying}. Bell inequalities can be defined as hyperplanes that bound the set of classical correlations expressed as a polytope in Euclidean space~\cite{Froissart1981Constructive,brunner2014bell}.
Explicitly characterizing this polytope and the set of quantum correlations that contains it
is already nontrivial in the minimal scenario with two parties, two inputs, and two outputs~\cite{Barizien2025Quantum}.
In fact, computing the maximum quantum violation, or \textbf{quantum value} $\omega_q$, of a Bell inequality is provably a difficult problem. There are known examples for which this is NP-hard~\cite{kempe2011entangled}, and in general the quantum value is uncomputable~\cite{ji2020mip}.

Nevertheless, there is a suite of numerical algorithms to compute bounds on $\omega_q$. One of the most notable general-purpose algorithms for computing lower bounds is the see-saw method~\cite{werner2001bell}, which iteratively optimizes the quantum measurement used by each party, ceteris paribus. Each iteration is a semidefinite program (SDP). As for upper bounds, the algorithm used is usually the Navascu\'es-Pironio-Ac\'in (NPA) hierarchy~\cite{navascues2008convergent}, which is a hierarchy of SDP relaxations of the optimization problem defining the quantum value. Each level of the hierarchy is a tighter upper bound, and in the limit of infinitely large levels, the hierarchy converges to the commuting-operator value.

Although these are general-purpose optimizers for Bell inequalities, they scale poorly with the size of the Bell inequality (number of parties $\Nparty$, inputs $\Nin$, or outputs $\Nout$). In the case of the see-saw method, we show in the Appendices that the computational complexity of each iteration is $\sim (\Nin \Nout)^\Nparty d^{2 \Nparty}$, where $d$ is the dimension of the quantum system held by each party. This becomes prohibitively expensive very quickly. In our field tests, we find that the see-saw method implemented on a CPU already requires about 300 seconds to optimize a bipartite Bell inequality with only a dozen inputs and outputs. Although this may be sufficient for small-scale theoretical studies, it is wholly inadequate for pushing the frontier of Bell nonlocality research.
Also, for Bell experiments for violating more complicated Bell inequalities, such numerical methods could prove inefficient for evaluating the feasibility of the experiment under photon loss~\cite{PhysRevA.47.R747,Xu2026Violating} and noise~\cite{CiscoQuantumSync2025}.
Furthermore, multi-agent coordination with restricted communication often involves complex scenarios such as high frequency trading~\cite{brandenburger2016team,szegedy2020systems,ding2024coordinating} and distributed systems~\cite{hasanpour2017quantum,da2025entanglement,arun2025faster}. There, the number of inputs (possible stock prices, possible data rates of a transmitter) and the number of outputs (possible trade orders, possible transmission channels) can be very large. Hence, for both pure and applied research, we need a more scalable general-purpose optimizer for Bell inequalities.

This has motivated a line of work on new optimizers. The approach in~\cite{deng2018machine} fixes the quantum measurements used and parameterizes the shared quantum state with a neural network, reducing the maximal violation with these measurements to a ground-state search. The authors of~\cite{bharti2019teach} cast the joint choice of state and measurements as a reinforcement-learning problem, sequentially generating their parameters from a stochastic policy conditioned on the choices made so far and optimizing the policy from a scalar Bell-inequality reward evaluated only after the full sequence is produced. They recover Tsirelson's bound for the Clauser--Horne--Shimony--Holt (CHSH) inequality and the optima of several multipartite Bell inequalities, at the cost of a comparatively long sequential optimization procedure. In~\cite{gardiner2026learning}, the authors optimize a differentiable parameterization of measurements within a policy-gradient loop, treating each game as a black-box oracle and assessing the computed quantum value against externally supplied reference values. In every case the search returns a single feasible strategy, and hence a lower bound on $\omega_q$, but at modest problem sizes. The above methods fix the measurements, carry a sparse or black-box training signal, or otherwise scale poorly with the size of the Bell inequality.

In this paper, we give a scalable optimizer for computing lower bounds on $\omega_q$ that uses batched gradient descent (BGD). This is a general-purpose optimizer and scales to problem sizes well beyond the reach of prior methods--- games with a dozen parties, a thousand inputs, or a thousand outputs. BGD parameterizes quantum measurements and the shared quantum state so that it always stays within the feasible set. The shared quantum state is a normalized complex vector, and each measurement basis is defined by a unitary expressed as the exponential of an anti-Hermitian matrix, followed by a stochastic assignment of basis vectors to outputs when the local dimension of the shared quantum state is not equal to the number of outputs. The search for the quantum value is then a first-order optimization over unconstrained parameters for the quantum state and unitaries, with gradients supplied by automatic differentiation. When stochastic post-processing is used, its probability-simplex parameters are projected after each update.

To evaluate the objective function, BGD contracts the Bell expression in a single batched tensor operation over the state amplitudes, the measurement operators, and the Bell expression coefficients. Thus, unlike the see-saw method, BGD never assembles the dense $d^{\Nparty} \times d^{\Nparty}$ Bell operator. This is the primary reason it can handle very large Bell inequalities. The same contraction advances a batch of independent random initializations together, so that many starting points explore the nonconvex optimization landscape in parallel. With this design, our method is naturally implemented by a GPU, which we use for some of our field tests.

Our paper is structured as follows. In~\Cref{sec:prelim}, we present preliminaries: definitions of nonlocal games, a mathematically equivalent formulation of Bell nonlocality, as well as existing optimizers for computing the \textbf{classical value} $\omega_c$ (largest value achievable classically) and quantum value $\omega_q$  of Bell inequalities. In~\Cref{sec:classical}, we discuss some numerical shortcuts and heuristic algorithms for computing the classical value. We do this by re-expressing the optimization as a mixed-integer linear program (MILP) and feeding it into a general-purpose MILP solver. We present our BGD optimizer in~\Cref{sec:bgd}. We then evaluate our optimizer and compare its performance with the see-saw method over a variety of Bell inequalities in~\Cref{sec:evaluation}. In particular, we use our BGD to evaluate how quantum entanglement can be used to increase the probability of identifying a common radio frequency band under adversarial jamming~\cite{wehner2026jamming}, a model of a real-world coordination problem, and demonstrate its scalability to instances with \(1820\) inputs. We also use BGD to optimize the Bell inequality tailored to maximally entangled states~\cite{Salavrakos2017Bell} with 1407 outputs.
We conclude with a discussion in~\Cref{sec:discussion}.
In~\Cref{app:seesaw-eng} to~\Cref{app:bgd-eng} we give further details of the different optimizers used in our numerical experiments.~\Cref{app:expressions} explains the different ways to express Bell inequalities. In~\Cref{app:multioutput-chsh}, we prove that a generalization of the CHSH game to more than two outputs (parties need to produce the same output or different outputs depending on the inputs) has a quantum value of $\cos^2 \frac \pi 8$ for all $\Nout \geq 2$. In~\Cref{app:cglmpgraph-instances}, we explicitly give the graphs that define the instances belonging to a class of Bell inequalities that we use as a benchmark for our optimizer. 

\section{Preliminaries}\label{sec:prelim}
\subsection{Nonlocal games}\label{section:nolocalgame}
We will use the mathematical framework of nonlocal games, which formalizes coordination tasks among non-communicating parties. In this framework, optimizing a Bell functional is equivalent to maximizing the expected utility of a game, as explained below.
In such a setting, multiple players participate in a joint task but are not allowed to communicate once the game begins. A schematic representation is shown in~\Cref{Fig:nonlocal_game}. For convenience, we replace the probabilistic predicate, which takes values in $[0,1]$, with a utility function that takes values in $\R$.

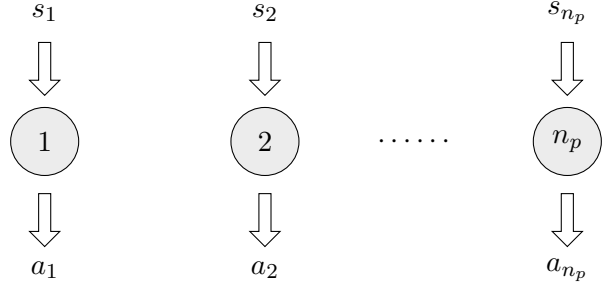
\begin{figure}[htbp!]
    \centering
    \begin{tikzpicture}[node distance=2cm,
  player/.style = {draw, circle, minimum width=.9cm, fill=gray!15},
  singlearrow/.style = {single arrow, draw, inner sep=0pt,
    minimum height=.7cm, minimum width=4mm, single arrow head extend=.3pt, rotate=-90}]

  \node[player] (P1) at (0,0) {$1$};
  \node[player, right=of P1] (P2) {$2$};        %
  \node[player] (Pn) at ($(P2)+(4,0)$) {$\Nparty$}; %

  \node at ($(P2)!0.5!(Pn)$) {$\cdots\cdots$};

  \node[singlearrow] at ([yshift=1cm]P1) {};
  \node[singlearrow] at ([yshift=1cm]P2) {};
  \node[singlearrow] at ([yshift=1cm]Pn) {};
  \node at ([yshift=1.7cm]P1) {$s_{1}$};
  \node at ([yshift=1.7cm]P2) {$s_{2}$};
  \node at ([yshift=1.7cm]Pn) {$s_{\Nparty}$};

  \node[singlearrow] at ([yshift=-1cm]P1) {};
  \node[singlearrow] at ([yshift=-1cm]P2) {};
  \node[singlearrow] at ([yshift=-1cm]Pn) {};
  \node at ([yshift=-1.7cm]P1) {$a_{1}$};
  \node at ([yshift=-1.7cm]P2) {$a_{2}$};
  \node at ([yshift=-1.7cm]Pn) {$a_{\Nparty}$};
\end{tikzpicture}
\caption{A nonlocal game with parties $1,\dots,\Nparty$. Each party $i$ receives input $s_i$ and produces output $a_i$.}
\label{Fig:nonlocal_game}
\end{figure}

\noindent We now give a formal definition.
\begin{dfn}
Let $\Nparty \ge 2$ be the number of parties. For each $i \in [\Nparty]$, let $S_i$ and $A_i$ be finite input and output sets. A \textbf{nonlocal game} is specified by:
\begin{itemize}
    \item a probability distribution $\pi$ over $\prod_{i=1}^{\Nparty} S_i$,
    \item a utility function
    \begin{equation*}
        \mathcal V : \prod_{i=1}^{\Nparty} A_i \times \prod_{i=1}^{\Nparty} S_i \to \R.
    \end{equation*}
\end{itemize}
We denote the game by $(\mathcal V, \pi)$.
\end{dfn}

In a single round of the game, each party $i$ is given an input $s_i$ sampled according to $\pi$, and produces an output $a_i$. The players are assumed to know $\pi$ and $\mathcal V$ in advance and may coordinate their strategy beforehand, but are not allowed to communicate after receiving their inputs.

Any strategy defines a conditional probability distribution
\begin{equation*}
p(\mathbf a \vert \mathbf s) \coloneqq p(a_1,\dots,a_{\Nparty} \vert s_1,\dots,s_{\Nparty}),
\end{equation*}
which we refer to as a \textbf{behavior}. This distribution describes how outputs are generated given inputs.
The corresponding \textbf{expected utility} is
\begin{equation*}
    \bar{\mathcal{V}} \coloneqq \sum_{a_i\in A_i, s_i\in S_i} \pi(\mathbf s) \cdot p(\mathbf a \vert \mathbf s) \mathcal V(\mathbf a \vert \mathbf s).
\end{equation*}

Define the \textbf{weighted utility tensor} by
\begin{equation}
    W_{\mathbf s,\mathbf a}
    \coloneqq \pi(\mathbf s)\,\mathcal V(\mathbf a\vert\mathbf s).
    \label{eq:weight-tensor}
\end{equation}
The expected utility is therefore a Bell functional,
\begin{equation*}
        \bar{\mathcal V}(p)
    =
    \sum_{\mathbf s\in\prod_i S_i}
    \sum_{\mathbf a\in\prod_i A_i}
    W_{\mathbf s,\mathbf a}\,p(\mathbf a\vert\mathbf s).
\end{equation*}
Conversely, any real coefficient tensor $W$ in a finite Bell scenario
can be represented in this form: choose an input distribution $\pi$
with full support and set
$\mathcal V(\mathbf a\vert\mathbf s)
= W_{\mathbf s,\mathbf a}/\pi(\mathbf s)$.
Thus, optimizing a Bell functional over a given class of behaviors
is equivalent to maximizing the expected utility of the corresponding
game over that class. This representation is generally not unique.

\paragraph{Classical strategies}
We begin with deterministic strategies.

\begin{dfn}
A \textbf{deterministic strategy} consists of functions $f_i : S_i \to A_i$ for each $i \in [\Nparty]$. The corresponding \textbf{deterministic behavior} is
\begin{equation*}
        p(\mathbf a \vert \mathbf s)
    = \prod_{i=1}^{\Nparty} \delta_{a_i, f_i(s_i)},
\end{equation*}
where $\mathbf a \in \prod_i A_i$ and $\mathbf s \in \prod_i S_i$.
\end{dfn}
\noindent More generally, \textbf{classical strategies} may employ shared randomness. Any such strategy can be expressed as a convex combination of deterministic ones. Thus, a \textbf{classical behavior} is a probabilistic mixture of deterministic behaviors. The highest possible expected utility achievable by classical strategies is denoted $\omega_c$:
\begin{equation*}
    \omega_c \coloneqq \max_\text{classical} \bar{\mathcal V}.
\end{equation*}
A \textbf{Bell inequality} is an upper bound on expected utility by $\omega_c$,
\begin{equation*}
    \bar{\mathcal V} \leq \omega_c,
\end{equation*}
which holds by definition of $\omega_c$ whenever $\bar{\mathcal V}$ is the expected utility of a classical strategy.

\paragraph{Quantum strategies}
We now allow the parties to share quantum entanglement. A quantum strategy for a nonlocal game is defined as follows.
\begin{dfn}
Let $B_i$ be a finite-dimensional Hilbert space for each $i \in [\Nparty]$. A \textbf{quantum strategy} consists of:
\begin{itemize}
    \item a shared pure state $\ket{\psi} \in \bigotimes_{i=1}^{\Nparty} B_i$,
    \item for each input $s_i$, a projective measurement
    \begin{equation*}
    M_i(s_i) = \{\Pi_{i,a_i}(s_i)\}_{a_i \in A_i}
    \end{equation*}
    acting on $B_i$, where the projectors satisfy
    \begin{equation*}
        \Pi_{i,a_i}(s_i) = \Pi_{i,a_i}(s_i)^\dagger, \qquad \Pi_{i,a_i}(s_i)\Pi_{i,a_i'}(s_i) = \delta_{a_i,a_i'}\Pi_{i,a_i}(s_i), \qquad \sum_{a_i \in A_i} \Pi_{i,a_i}(s_i) = I_{B_i}.
    \end{equation*}
\end{itemize}
The resulting \textbf{quantum behavior} is
\begin{equation*}
    p(\mathbf a \vert \mathbf s)
    = \bra{\psi}
    \bigotimes_{i=1}^{\Nparty} \Pi_{i,a_i}(s_i)
    \ket{\psi}.
\end{equation*}
\end{dfn}
\noindent The quantum value of the game, denoted by $\omega_q$, is the supremum of the expected utility over all finite-dimensional quantum strategies:
\begin{equation*}
 \omega_q \coloneqq \sup_\text{quantum} \bar{\mathcal V}.
\end{equation*}
The supremum ranges over all finite local dimensions and need not
be attained by any finite-dimensional strategy~\cite{slofstra2019set}.

Since $\omega_q$ ranges over arbitrary finite local dimensions,
restricting to pure states and projective measurements entails no loss
of generality. Any
mixed shared state can be purified by enlarging one party's local
Hilbert space without changing the resulting behavior.
And by Naimark's theorem, any POVM can be realized as a projective measurement on a larger space, a construction also used experimentally~\cite{Conlon2023Discriminating}.
At a fixed local dimension, however, this equivalence need not hold; for example, in the two-qubit inequality of~\cite{vertesi2010two}, a qubit POVM outperforms every qubit projective measurement. Because the available measurement resources are platform-dependent and we do not model a specific physical implementation, we adopt rank-one projective measurements as BGD's default search class. The see-saw method~\cite{werner2001bell,P_l_2010} instead optimizes general POVMs at the chosen dimension. Only for Mutually unbiased bases (MUB) does BGD additionally optimize a stochastic relabeling of the outcomes (classical post-processing); the resulting effects commute for each input, so this construction does not parameterize arbitrary POVMs at that dimension.

Classical behaviors form a subset of quantum behaviors. There exist nonlocal games for which quantum strategies outperform all classical ones:
\begin{equation*}
 \omega_q > \omega_c.
\end{equation*}
This separation indicates a \textbf{Bell inequality violation} or a \textbf{quantum advantage}: a strategy whose expected utility exceeds the classical value, $\bar{\mathcal V} > \omega_c$.

\subsection{Existing optimizers for the classical value}
Since a classical behavior is a convex combination of deterministic behaviors and the expected utility is linear in the behavior, to compute $\omega_c$ it is sufficient to perform a finite optimization over deterministic strategies.
This observation underlies the exact classical optimizers reviewed below.

\paragraph{Brute force enumeration.}
The most direct optimizer enumerates all deterministic local strategies and evaluates each. A party with $\abs{S_i}$ inputs and $\abs{A_i}$ outputs has $\abs{A_i}^{\abs{S_i}}$ possible functions, so the joint strategies number $$\prod_{i=1}^{\Nparty} \abs{A_i}^{\abs{S_i}}.$$
Setting $\abs{S_i} = \Nin$ and $\abs{A_i} = \Nout$ for all $i \in [\Nparty]$ (we will call this the \textbf{symmetric case}), the total number of deterministic strategies is
\begin{equation*}
    N_\mathrm{det} = \Nout^{\Nparty \Nin}.
\end{equation*}
This number grows exponentially with the number of parties and inputs, which sit in the exponent, and polynomially (of degree $\Nparty\Nin$) with the number of outputs; brute force is therefore infeasible outside small Bell scenarios. %

\paragraph{LP-based methods.}
A second type of optimizer exploits the convex geometry of the classical set~\cite{zukowski1999strengthening,kaszlikowski2000violations}. Let $d_\lambda$ enumerate the deterministic behaviors.
The classical set is the polytope
$\mathcal L=\operatorname{conv}\{d_\lambda\}$,
whose vertices are precisely the deterministic behaviors. We view a behavior as a real vector in the ambient space
\begin{equation*}
    \R^{\prod_{i=1}^{\Nparty} \abs{A_i}\abs{S_i}}.
\end{equation*}
Maximizing the expected utility over $\mathcal L$ is a linear program in the vertex weights $\{q_\lambda\}$, with $q_\lambda\geq0$ and $\sum_\lambda q_\lambda=1$. Deciding whether a given behavior is classical is the corresponding feasibility problem: a feasible solution gives an explicit local decomposition, whereas a dual infeasibility certificate obtained through LP duality gives a separating Bell inequality. The formulation has one variable for each deterministic behavior, of which there are
$\prod_i\abs{A_i}^{\abs{S_i}}$, and hence grows exponentially with the number of parties or inputs
for a fixed number of outputs greater than one.

\paragraph{Tropical tensor-network methods.}
A more recent direction exploits structure in the Bell inequality rather than treating the classical optimization as a generic combinatorial search. When the bound can be written as a minimization of a sum of local cost terms, it can be recast as a tensor-network contraction carried out in the tropical algebra $(\min,+)$ and evaluated by dynamic programming~\cite{Hu2026Tropical}. This is highly effective for inequalities with few-body correlators and finite interaction range on geometries that admit a low-width contraction, including one-dimensional geometries. For some translationally invariant one-dimensional systems, the method also extends to the thermodynamic limit through tropical renormalization. However, exact contraction can require exponentially large
intermediate tensors, so the method does not guarantee efficient
optimization for arbitrary nonlocal games.

\subsection{Existing optimizers for the quantum value}
\label{subsec:existing_quantum}
Let $\mathcal Q$ denote the set of quantum behaviors realizable over arbitrary finite local dimensions. This set is convex, and $\omega_q=\sup_{p\in\mathcal Q}\bar{\mathcal V}(p)$, where $\bar{\mathcal V}$ is linear in $p$. The difficulty is that $\mathcal Q$ admits no known tractable characterization~\cite{navascues2008convergent,slofstra2019set,miller2021membership,ji2020mip},
so no tractable exact convex formulation is known in general.
Instead, numerical methods construct explicit strategies to lower-bound $\omega_q$ or optimize over outer relaxations to upper-bound it. %

Recall that a quantum strategy consists of a shared state $\ket{\psi}$ together with local measurements $\{\Pi_{i,a_i}(s_i)\}$, realizing the behavior $p(\mathbf a\vert\mathbf s)=\bra{\psi}\bigotimes_{i=1}^{\Nparty} \Pi_{i,a_i}(s_i)\ket{\psi}$. For fixed local measurements, define the \textbf{Bell operator}
\begin{equation}
    \mathcal B \coloneqq \sum_{\mathbf s\in\prod_i S_i}\sum_{\mathbf a\in\prod_i A_i} \pi(\mathbf s)\, \mathcal V(\mathbf a\vert \mathbf s)\, \bigotimes_{i=1}^{\Nparty} \Pi_{i,a_i}(s_i).
    \label{eq:bell-operator}
\end{equation}
Then $\bar{\mathcal V} = \bra{\psi}\mathcal B\ket{\psi}$.
The same expression applies to general POVMs, with the projectors
replaced by the corresponding measurement effects.

\paragraph{See-saw method.}
The see-saw method~\cite{werner2001bell,P_l_2010} is a heuristic for lower bounds on~$\omega_q$ that alternates optimization over the components of a quantum strategy. It fixes all but one component---a single party's collection of POVMs or the shared state---optimizes that component, and iterates. Each measurement update can be formulated as a semidefinite program (SDP), while the state update is a largest-eigenvalue problem. Solving each subproblem exactly makes the expected utility nondecreasing. Because the joint problem is nonconvex, however, the method is not guaranteed to reach the global optimum, and its result can depend on the initialization. In practice, the method is run from several random initializations, and the best result is returned.

For simplicity, we assume all quantum systems $B_i$ have the same fixed dimension $d$.
Fixing the state $\rho=\ketbra{\psi}$ and all POVMs but party~$k$'s makes the functional linear in party~$k$'s measurement effects $\{M_{a_k\vert s_k}\}_{s_k,a_k}$. Thus this subproblem is the SDP
\begin{equation}
    \max_{\{M_{a_k\vert s_k}\}} \sum_{s_k,a_k}\tr\!\big(\Gamma^{(k)}_{s_k,a_k}\, M_{a_k\vert s_k}\big)
    \quad\text{s.t.}\quad
    M_{a_k\vert s_k}\succeq 0 \quad\forall s_k,a_k,\ \ \sum_{a_k}M_{a_k\vert s_k}=\id_d \quad\forall s_k,
    \label{eq:seesaw_povm_sdp}
\end{equation}
where $\Gamma^{(k)}_{s_k,a_k}$ is the $d\times d$ Hermitian matrix obtained by contracting the Bell coefficients and $\rho$ with the other parties' fixed POVMs for each $(s_k,a_k)$. Multipartite see-saw optimization is supported by the Julia \texttt{Ket.jl} library~\cite{KetJl}. We use a modified copy of this implementation; the modifications to the state update and convergence criterion are described in~\Cref{app:seesaw-eng}. Its per-sweep complexity is derived in~\Cref{app:bigo-seesaw}.

\paragraph{Black-box optimization.}
A projective measurement can be parameterized by an orthonormal basis, represented by a unitary, and an assignment of its basis vectors to outcome subspaces~\cite{ding2024coordinating}. For fixed measurements, maximizing over the shared state gives the largest eigenvalue of the Bell operator, with a corresponding top eigenvector. Basis-vector phases are irrelevant to the projectors, and for each party one input basis can be fixed to the computational basis by a local change of basis.

Previous work treated this eigenvalue as a black-box function of the measurement parameters. It used derivative-free grid search and the CMA-ES evolutionary search~\cite{hansen2016cma,hamano2022cma,nomura2024cmaes}, and also considered gradients estimated by finite differences. With $11$ grid points per continuous parameter, and including the discrete assignments to outcomes, two-qubit CHSH requires on the order of $10^6$ Bell-operator evaluations; the analogous three-party scenario exceeds $10^9$. CMA-ES avoids a fixed grid but still relies on sampled function values, while a finite-difference gradient requires additional evaluations for each parameter. These costs motivate computing gradients directly by automatic differentiation in our quantum optimizer.

\paragraph{NPA hierarchy.}
We also outline an algorithm used to compute upper bounds on $\omega_q$.

The Navascu\'es--Pironio--Ac\'{\i}n (NPA)
hierarchy~\cite{navascues2008convergent,navascues2007bounding}
provides upper bounds on $\omega_q$ without restricting the quantum dimension. Each level optimizes over an outer
relaxation of the quantum set using semidefinite programming.
Increasing the level gives non-increasing upper bounds, at the
cost of larger optimization problems. These bounds complement
the lower bounds obtained by the see-saw
method and the BGD optimizer; when the bounds coincide, the quantum value
is certified. In general, however, the hierarchy converges to
the commuting-operator value, which can exceed the
finite-dimensional tensor-product value
$\omega_q$~\cite{ji2020mip}.

For multipartite games, we use an augmented first-level relaxation
that involves every term in the Bell expression.
Its construction, implementation, and computational complexity
are detailed in~\Cref{app:npa-eng}. We mention that modifications to the vanilla NPA algorithm were proposed to improve scalability~\cite{mortimer2025bounding}, but the resulting algorithm may not converge to the commuting-operator value.

\section{Improved Classical Optimizers}\label{sec:classical}
In this paper we compute the classical value $\omega_c$ using two complementary methods for the same finite optimization problem: a vectorized brute-force search that runs on a GPU and a mixed-integer linear program (MILP) solved using a commercial global solver. Both target the same optimum; they differ in the regimes in which they are efficient. Throughout we keep the notation of \Cref{section:nolocalgame}: each party $i\in[\Nparty]$ receives an input $s_i\in S_i$ and returns an output $a_i\in A_i$, the inputs are drawn from a prior $\pi$, and the utility is $\mathcal V$. For complexity statements we assume $\abs{S_i}=\Nin$ and $\abs{A_i}=\Nout$ for all $i$, so the number of deterministic strategies is
\begin{equation*}
N_{\mathrm\det}=\Nout^{\Nparty\Nin}.
\end{equation*}
Since every classical behavior is a convex combination of deterministic local
behaviors and the expected utility is linear, it suffices to optimize over
deterministic local strategies.

\subsection{Vectorized GPU brute force}
The most direct exact method would enumerate all deterministic local
strategies. We reduce this enumeration using one party's best response and
evaluate the remaining strategies in parallel on a GPU.

\paragraph{Best response for one party.}
A deterministic strategy is a tuple of functions \(f_i:S_i\to A_i\), one for
each party \(i\), so in the general, possibly asymmetric case their number is
\begin{equation*}
N_{\det}=\prod_{i=1}^{\Nparty}|A_i|^{|S_i|}.
\end{equation*}

We eliminate one party, denoted by \(b\), by optimizing its best response rather
than enumerating its response function.
For fixed functions \(f_i:S_i\to A_i\) of all parties \(i\ne b\), write
\begin{equation*}
f_{b^c}:=\{f_i\}_{i\ne b},
\qquad
s_{b^c}:=(s_i)_{i\ne b}.
\end{equation*}
For each input \(s_b\in S_b\) and each candidate output \(a_b\in A_b\),
define the score
\begin{equation*}
C_{f_{b^c}}(s_b,a_b)
:=
\sum_{s_{b^c}\in \prod_{i\ne b}S_i}
\pi(s_b,s_{b^c})\,
\mathcal V\!\left(a_b,\{f_i(s_i)\}_{i\ne b}\vert s_b,s_{b^c}\right).
\end{equation*}
Then party \(b\)'s optimal response to \(f_{b^c}\) is chosen independently
for each input \(s_b\):
\begin{equation*}
f_b^*(s_b)\in \arg\max_{a_b\in A_b} C_{f_{b^c}}(s_b,a_b).
\end{equation*}
The maximization separates over \(s_b\) because the choice of \(f_b(s_b)\)
affects only terms with that input; there is no constraint coupling party
\(b\)'s outputs at different inputs.
Therefore the classical value can be computed as
\begin{equation*}
\omega_c
=
\max_{\{f_i:S_i\to A_i\}_{i\ne b}}
\sum_{s_b\in S_b}
\max_{a_b\in A_b}
C_{f_{b^c}}(s_b,a_b).
\end{equation*}

This replaces the enumeration of party \(b\)'s
\(|A_b|^{|S_b|}\) possible functions by a per-input maximization over
\(|A_b|\) outputs.  For fixed \((s_b,a_b)\), the sum defining
$C_{f_{b^c}}(s_b,a_b)$ contains one term for each of the
$\prod_{i\neq b}|S_i|$ choices of $s_{b^c}$.\footnote{For simplicity, we assume
in this paragraph that $\pi$ has full support on $\prod_i S_i$.} Thus the total
number of terms evaluated is
\begin{equation}
\label{eq:N_eval}
N_\mathrm{eval}
=
\left(\prod_{i\ne b}|A_i|^{|S_i|}\right)
|A_b|
\prod_{i=1}^{\Nparty}|S_i|
=
N_{\det}\,
\frac{\prod_{i=1}^{\Nparty}|S_i|}{|A_b|^{|S_b|-1}}.
\end{equation}
To minimize the number of terms, we choose
\begin{equation*}
b\in \arg\max_i |A_i|^{|S_i|-1}.
\end{equation*}
In the symmetric case where \(|S_i|=\Nin\) and
\(|A_i|=\Nout\), this becomes
\begin{equation*}
N_\mathrm{eval} = \Nout^{(\Nparty-1)\Nin+1}\Nin^{\Nparty}.
\end{equation*}
Na\"ive brute force instead evaluates
$\Nout^{\Nin\Nparty}\Nin^{\Nparty}$ weighted-utility
terms, so the best-response reduction saves a factor of
$\Nout^{\Nin-1}$.

\paragraph{Dense batched evaluation.}
After eliminating party \(b\), there remain
$\prod_{i\ne b}|A_i|^{|S_i|}$ partial deterministic strategies, or
$\Nout^{\Nin(\Nparty-1)}$ in the symmetric case. Since these
strategies can be evaluated independently, we process them in parallel batches
on a GPU. Each batch is represented by a contiguous block of integer strategy
labels, which are decoded by mixed-radix decomposition to generate
$f_i(s_i)$, $i\ne b$, on the fly. Thus the full table of partial strategies
need not be stored.

The resulting outputs select entries of the weighted utility tensor
$W_{\mathbf s,\mathbf a}$. Summing over
$s_{b^c}$ produces $C_{f_{b^c}}(s_b,a_b)$, and party $b$'s best response is the
reduction
\begin{equation*}
\sum_{s_b\in S_b}\max_{a_b\in A_b} C_{f_{b^c}}(s_b,a_b).
\end{equation*}
Decoding, indexing, summation, and maximization are performed as dense tensor
operations. If the full enumeration exceeds GPU memory, we cap the batch size
and retain the largest value across batches. Batching therefore changes only
the storage and execution of the search, not the set of strategies evaluated.

\subsection{Mixed-integer linear program representation}\label{subsec:milp}
When the input cardinality $\Nin$ is large, the enumeration over $\Nout^{\Nin(\Nparty-1)}$ partial strategies becomes infeasible.  We instead encode the output chosen for every party and input as a mixed-integer linear program (MILP).  A branch-and-bound solver repeatedly fixes some of these choices and uses a relaxed problem to upper-bound all strategies consistent with the remaining choices.  Any branch whose upper bound cannot improve the best strategy already found can be discarded without enumerating its individual strategies.

\paragraph{One-hot strategies without binaries.} Represent party $i$'s local
response to input $s_i$ by the variables
\begin{equation}
    P_{i,s_i,a_i} \coloneqq p_i(a_i\vert s_i)\in[0,1],
    \qquad
    \sum_{a_i\in A_i} P_{i,s_i,a_i}=1.
    \label{eq:milp_simplex}
\end{equation}
For each party-input pair $(i,s_i)$, exactly one output must be selected.  We
enforce this using a special ordered set of type~1 (SOS1): at most one variable
in $\{P_{i,s_i,a_i}\}_{a_i\in A_i}$ may be nonzero.  Together with the
normalization in \Cref{eq:milp_simplex}, this makes that variable equal to one
and all others zero, so $P_{i,s_i,\cdot}$ represents a deterministic response.
The feasible discrete assignments are therefore in one-to-one correspondence
with the $\prod_i|A_i|^{|S_i|}$ deterministic strategies.  Although effective
relaxation bounds may eliminate many strategies, the worst-case work can remain
comparable to explicit enumeration.

\paragraph{Linearized objective via a joint selector.} A second block of variables $t_{\mathbf s,\mathbf a}\in[0,1]$, one for each pair of an input combination $\mathbf s\in\prod_iS_i$ and an output combination $\mathbf a\in\prod_iA_i$, represents whether the joint output $\mathbf a$ is selected for input $\mathbf s$.  The following linear consistency constraints tie these variables to the parties' local choices:
\begin{equation}
    \sum_{\mathbf a:\,a_i=o} t_{\mathbf s,\mathbf a} \;=\; P_{i,s_i,o}
    \qquad \forall\, i\in[\Nparty],\ \mathbf s\in\textstyle\prod_iS_i,\ o\in A_i.
    \label{eq:milp_marginal}
\end{equation}
Once the local choices $P$ are one-hot, \Cref{eq:milp_marginal} forces
$t_{\mathbf s,\mathbf a}=1$ for the unique joint output prescribed by the
strategy and zero for every other output.  Equivalently,
$t_{\mathbf s,\mathbf a}=\prod_iP_{i,s_i,a_i}$ at every deterministic
solution.  The constraints also imply
$\sum_{\mathbf a}t_{\mathbf s,\mathbf a}=1$, so no separate normalization is
needed.  The classical value is therefore
\begin{equation}
    \omega_c \;=\; \max \ \sum_{\mathbf s,\mathbf a} \pi(\mathbf s)\,\mathcal V(\mathbf a\vert \mathbf s)\, t_{\mathbf s,\mathbf a}
    \label{eq:milp_objective}
\end{equation}
subject to \Cref{eq:milp_simplex}, the SOS1 sets, and \Cref{eq:milp_marginal}. The constraint matrix depends only on the game \emph{shape} $(\abs{S_i},\abs{A_i})$; a particular game enters solely through the objective coefficients $\pi(\mathbf s)\,\mathcal V(\mathbf a\vert \mathbf s)$.

\paragraph{Optimality bounds and model size.} The objective and consistency
constraints are linear; the only discrete restrictions are the SOS1 output
choices.  We solve the resulting MILP with Gurobi.  Subject to the
solver's prescribed numerical tolerances, any feasible
incumbent gives a lower bound on $\omega_c$, while the solver's best bound gives
an upper bound.  Together they bracket $\omega_c$, and their difference is the
absolute optimality gap.  The model size is
dominated by the selector block: in
the symmetric case it has $(\Nin\Nout)^{\Nparty}$ variables $t_{\mathbf s,\mathbf a}$ and
$\Nparty\,\Nin^{\Nparty}\Nout$ marginal-consistency rows.

\subsection{Comparing the optimizers}
\label{subsec:classical_comparison}
The vectorized GPU brute-force optimizer is well suited to scenarios where each party's
strategy count $\abs{A_i}^{\abs{S_i}}$ is small, as occurs when the input
cardinality is small and the output cardinality remains moderate.  A
representative case is the MABK family (\Cref{section:mabk}), with two inputs
and outputs per party: each party has only $\Nout^{\Nin}=4$ response functions.
Although the remaining enumeration
$\Nout^{\Nin(\Nparty-1)}=4^{\,\Nparty-1}$ still grows exponentially with the
number of parties, the partial strategies can be evaluated independently and
in parallel on a GPU. They are generated batch by batch, so the full set need
not be stored at once.

For fixed $\Nparty$ and $\Nout$, the number of partial strategies enumerated by
brute force grows exponentially with $\Nin$, whereas the dominant MILP
selector block grows polynomially as $(\Nin\Nout)^{\Nparty}$. This difference helps explain
why the MILP is often preferable at large input cardinality, for example in
games tailored to maximally entangled states~\cite{Salavrakos2017Bell} with large parameters, although
it does not imply polynomial solution time. When the relaxation is
sufficiently tight, its upper bounds can exclude subsets of deterministic
assignments without evaluating them individually; the resulting advantage is
instance dependent. For fixed $\Nin$ and $\Nout$, both the brute-force enumeration
count and the MILP model size grow
exponentially with $\Nparty$. The brute-force search can stream independent
evaluations in parallel, while the full MILP model must be stored. Thus the
two methods trade direct parallel enumeration against bound-based elimination
of strategies.

\section{Batched Gradient Descent Optimizer}
\label{sec:bgd}

The gradient-free optimizer in~\Cref{subsec:existing_quantum} treats the Bell value as a black-box objective and therefore does not use the derivatives available from a differentiable parameterization. Batched gradient descent (BGD) instead applies first-order optimization to a parameterized family of feasible quantum strategies. The expected utility and its gradient are evaluated by a direct tensor contraction without ever constructing the full Bell operator. In the direct projective parameterization, feasibility is enforced by normalization and unitarity; when stochastic post-processing is used, its simplex constraints are enforced by projection after each update. Every evaluated strategy is therefore legitimate, so each value BGD reports is a lower bound on~$\omega_q$. We give here a high-level overview of the algorithm and provide further details in~\Cref{app:bgd-eng}.

\subsection{Differentiable quantum strategies}
\label{subsec:diff-quantum-strategy}

Fix local Hilbert-space dimensions $d_i$ and write $D=\prod_{i=1}^\Nparty d_i$. We optimize over shared pure states; this entails no loss for fixed measurements and dimensions because the expected utility is linear in the density operator and is therefore maximized at a pure state. The shared state is represented by a nonzero vector $z\in\C^D$, normalized before every objective evaluation,
\begin{equation}
    \ket{\psi_z}=\frac{z}{\norm{z}}.
    \label{eq:normalized-state}
\end{equation}
For each party $i$ and input $s_i\in S_i$, the measurement basis is generated by a differentiable unitary $U^{(i)}_{s_i}\in\mathrm U(d_i)$ generated from
an unconstrained complex matrix $A^{(i)}_{s_i} \in \mathbb C^{d_i\times d_i}$:
\begin{equation}
    U^{(i)}_{s_i}
    =
    \exp\!\left(\frac{A^{(i)}_{s_i}-(A^{(i)}_{s_i})^\dagger}{2}\right).
    \label{eq:matrix-exp-unitary}
\end{equation}
Since the exponent is skew-Hermitian, this construction gives a unitary for every $A^{(i)}_{s_i}$ and requires no projection onto $\mathrm U(d_i)$ after a gradient update.

We first describe the complete projective measurement used when the local dimension equals the number of outputs, $d_i=\abs{A_i}$. Each basis vector of $U^{(i)}_{s_i}$ is then assigned to one output. Writing $\ket{u^{(i)}_{s_i,a_i}}$ for the $a_i$-th column, the corresponding measurement element is the rank-one projector
\begin{equation}
    \Pi_{i,a_i}(s_i)
    =
    \ketbra{u^{(i)}_{s_i,a_i}}.
    \label{eq:local-projector}
\end{equation}
The trainable parameters are the state vector $z$ and the matrices $A=\{A^{(i)}_{s_i}\}$ that generate the unitaries via~\Cref{eq:matrix-exp-unitary}. %
The behavior generated by $\Theta=(z,A)$ is
\begin{equation}
    p_\Theta(\mathbf a\vert \mathbf s)
    =
    \bra{\psi_z}
    \left(
        \bigotimes_{i=1}^{\Nparty} \Pi_{i,a_i}(s_i)
    \right)
    \ket{\psi_z},
    \label{eq:parametric-quantum-behavior}
\end{equation}
and the objective is the expected utility
\begin{equation}
    J(\Theta)
    =
    \sum_{\mathbf s\in\prod_i S_i}
    \sum_{\mathbf a\in\prod_i A_i}
    \pi(\mathbf s)\,\mathcal V(\mathbf a\vert \mathbf s)\,p_\Theta(\mathbf a\vert \mathbf s).
    \label{eq:grad-quantum-objective}
\end{equation}
Every value returned by this optimization is attained by an explicit strategy in this variational family and hence lower-bounds~$\omega_q$; the parameterization need not exhaust all measurements available at the chosen dimensions.

When $d_i \neq \abs{A_i}$, we perform a rank-one projective measurement in the basis defined by $U_{s_i}^{(i)}$ followed by a stochastic map. This maps the $d_i$ measurement outcomes to the $\abs{A_i}$ outputs. For each party and input this map is a $d_i \times \vert A_i \vert$ row-stochastic matrix $(S^{(i)}_{s_i})_{b,a_i}\ge0$, $\sum_{a_i}(S^{(i)}_{s_i})_{b,a_i}=1$, and gives the POVM effects
\begin{equation*}
M^{(i)}_{a_i\vert s_i}=\sum_{b=1}^{d_i} (S^{(i)}_{s_i})_{b,a_i}\ketbra{u^{(i)}_{s_i,b}}.
\end{equation*}
For each fixed input these effects commute, since they arise by classical post-processing of a single projective measurement; this construction does not parameterize an arbitrary POVM at fixed dimension. The entries of all matrices $S^{(i)}_{s_i}$ are optimized jointly with the state and the unitaries. After every gradient step, each row is projected back onto the probability simplex; this preserves the row-stochastic constraints and permits boundary points, including deterministic $0/1$ assignments. When $d_i=\abs{A_i}$ and $S^{(i)}_{s_i}=I_{d_i}$ for every input $s_i$, this reduces to the projective measurement in~\Cref{eq:local-projector}. When $d_i > \abs{A_i}$ and $(S_{s_i}^{(i)})_{b,a_i} \in \{0,1\}$, this yields a
coarse-grained projective measurement, possibly with
higher-rank or zero projectors.

\subsection{Direct Bell-value contraction}
\label{subsec:direct-bell-contraction}

Substituting \Cref{eq:local-projector,eq:parametric-quantum-behavior}
into~\Cref{eq:grad-quantum-objective} and using the weighted utility
tensor $W$ defined in~\Cref{eq:weight-tensor} gives a direct
tensor-contraction expression for the objective. Using local basis
indices $r_i,t_i\in[d_i]$ and the convention
$U^{(i)}_{s_i;r_i a_i}\coloneqq
\langle r_i\vert u^{(i)}_{s_i,a_i}\rangle$, we obtain
\begin{equation}
\begin{aligned}
    J(\Theta)
    =
    \sum_{\mathbf s,\mathbf a,\mathbf r,\mathbf t}
    &
    W_{\mathbf s,\mathbf a}\,
    \overline{\psi_z(\mathbf r)}\,
    \psi_z(\mathbf t)
    \times
    \prod_{i=1}^{\Nparty}
    U^{(i)}_{s_i;r_i a_i}\,
    \overline{U^{(i)}_{s_i;t_i a_i}}.   
\end{aligned}
    \label{eq:bell-einsum}
\end{equation}
An equivalent evaluation first contracts $W$ with the local projectors $\Pi_{i,a_i}(s_i)$ to form the dense $D\times D$ Bell operator~\Cref{eq:bell-operator}, and then takes its expectation in the state. This familiar formulation requires storing an operator with $D^2$ entries. We instead contract~\Cref{eq:bell-einsum} by successively applying the local basis changes to the state indices before joining the bra and ket contractions through $W$. This order never forms the dense Bell operator. The contraction order and its cost are analyzed in~\Cref{app:cost}. Avoiding the explicit Bell operator is central to the improved scaling with quantum dimension and party number relative to the gradient-free optimizer of~\Cref{sec:prelim}.

When $d_i \neq \abs{A_i}$, the basis-outcome indices $b_i\in[d_i]$ must be retained and then mapped to the reported outputs by $S^{(i)}_{s_i}$. In this case the corresponding contraction is
\begin{equation}
\label{eq:stochastic}
\begin{aligned}
    J(\Theta)
    =
    \sum_{\mathbf s,\mathbf a,\mathbf b,\mathbf r,\mathbf t}
    &
    W_{\mathbf s,\mathbf a}\,
    \overline{\psi_z(\mathbf r)}\,
    \psi_z(\mathbf t)
    \times
    \prod_{i=1}^{\Nparty}
    U^{(i)}_{s_i;r_i b_i}\,
    (S^{(i)}_{s_i})_{b_i,a_i}\,
    \overline{U^{(i)}_{s_i;t_i b_i}}.  
\end{aligned}
\end{equation}

\subsection{Extensions}
\label{subsec:extensions}

The contraction \Cref{eq:bell-einsum} admits the following extensions.

\paragraph{Batched restarts.} Since the optimization landscape is nonconvex, we use many random initializations. We optimize a batch of $q$ of them together by adding a leading batch index to every tensor in~\Cref{eq:bell-einsum} except tensor $W$, which is shared across the batch. A single contraction then returns $q$ objective values, and all $q$ strategies are advanced in parallel; the best batch element encountered is retained. On parallel tensor hardware the per-iteration work is dominated by batched matrix multiplications, so the marginal cost of additional restarts remains small until the hardware is saturated (see \Cref{app:batching}).

\paragraph{Fixed state.} For controlled comparisons in which different methods must be evaluated on the same shared state, the state $z$ is supplied in advance and held fixed, optimizing only the local measurements. A common choice when every quantum system has the same dimension is the GHZ-type state
\begin{equation}
    \ket{\Phi_d^{(\Nparty)}}=\frac{1}{\sqrt d}\sum_{j=0}^{d-1}\ket{j}^{\ot \Nparty}.
    \label{eq:fixed-ghz-state}
\end{equation}

\subsection{Objective and optimization}
\label{subsec:objective-optimization}

BGD minimizes the loss $\mathcal L(\Theta)=-J(\Theta)$, equivalently maximizing the expected utility~\Cref{eq:grad-quantum-objective}. When batching is used, the loss is summed over the restarts; because their parameters are disjoint, automatic differentiation supplies the full gradient for each strategy independently, while the per-restart conditioning keeps their step scales separate. The state vector and unitary generators are updated directly, while the stochastic post-processing matrices, when present, are projected onto their probability simplices as described above.

Two numerical choices make the optimization robust across games with different utility scales. First, Adam's coordinate-wise preconditioning is supplemented by a per-restart scalar factor derived from its second-moment estimate. This factor leaves the step unchanged during the initial optimization regime and damps it near stationary points. Second, an adaptive learning-rate procedure decreases the step when progress stalls and restores both the parameters and Adam moments when a trial phase is rejected. These mechanisms are detailed in \Cref{app:preconditioning} and \Cref{app:linesearch}. The search terminates when the relative change in the objective stays below a tolerance for a fixed patience window, or when a time or iteration budget is reached, and returns the best feasible strategy encountered.

\section{Evaluation}\label{sec:evaluation}
We benchmark families of Bell inequalities
across a range of \textbf{game shapes} $(\Nparty,\Nin,\Nout)$.
These families are from existing Bell nonlocality literature, are motivated by real-world multi-agent coordination problems, or are generated randomly.
We organize the benchmarks into solved and unsolved families.

For \textbf{solved families}, the classical and quantum values are known analytically for every instance. The analytic quantum value therefore provides an exact reference for assessing the accuracy and convergence of BGD and the see-saw method. We characterize the numerical reach of each method by varying the number of parties $\Nparty$, the number of inputs $\Nin$, the number of outputs $\Nout$, or the local quantum dimension $d$ along prescribed scaling axes. Solved family benchmarks include the Bell inequalities tailored to maximally entangled states~\cite{Salavrakos2017Bell}, the MABK inequalities~\cite{Mermin1990a,Ardehali1992,BelinskyKlyshko1993a}, the Mermin inequalities~\cite{Mermin1990a}, and the chained Bell inequalities~\cite{Pearle1970,Braunstein1990Wringing}.

For \textbf{unsolved families}, the exact quantum value is not known in general, and in some cases the exact classical value is also unknown. For the families used to directly compare optimizers, BGD and the see-saw method provide lower bounds on $\omega_q$ through explicit feasible strategies, while the NPA hierarchy provides an upper bound whenever the corresponding SDP relaxation can be solved within the specified computational resource limit.
Unsolved family benchmarks include the CGLMP inequalities~\cite{CollinsEtAl2002a}, load-balancing games~\cite{hasanpour2017quantum}, jamming games~\cite{wehner2026jamming}, random CGLMP networks, and uniformly random games. We report selected instances to compare the numerical methods and bounds.

Separately, for the mutually unbiased bases
inequalities~\cite{Colomer2022threenumerical}, we run BGD
with a fixed quantum state at parameter choices for which
that state is known to attain the analytic quantum optimum,
and compare with published numerical results where available.

\subsection{Experimental setup}\label{subsec:eval-setup}
\paragraph{Evaluation protocol.}
Unless stated otherwise, CPU-based BGD, GPU-based BGD, and the see-saw method each perform 10 independently seeded optimization restarts. The three method/device configurations use separate random initializations, so the CPU and GPU BGD restarts do not share initial parameters with each other or with the see-saw restarts. A restart is successful if it reaches the effective tolerance within the $300\,\mathrm{s}$ budget; otherwise it is unsuccessful. For the tabulated optimizer comparisons, BGD uses batch size~$1$. BGD also supports larger batches, which can substantially reduce the wall-clock time for a fixed number of restarts by amortizing fixed tensor-operation launch costs across concurrent initializations (see~\Cref{app:batching}). We retain batch size~$1$ in the tabulated optimizer comparisons so that each restart has a directly comparable time budget and the comparison reflects per-restart optimizer behavior rather than an advantage from parallel execution unavailable to the see-saw method.

The \textbf{success rate} is the fraction of these restarts that reach the effective tolerance before either the convergence criterion or the time limit; the effective tolerance and the corresponding success criterion are defined below.

\paragraph{Reference values and success criterion.}
ach optimizer restart is evaluated against a fixed quantum reference value:
\begin{equation*}
    q_{\mathrm{ref}}=
    \begin{cases}
        \omega_q, & \text{for solved families},\\
        q_{\downarrow}, & \text{for unsolved families}.
    \end{cases}
\end{equation*}
For an unsolved family, $q_{\downarrow}\leq\omega_q$ is the largest feasible value found in 256 independently seeded GPU BGD restarts and serves as a numerical proxy for $\omega_q$. The BGD values in~\Cref{tab:eval-mub} are also maxima over 256 such restarts. These reference calculations use larger batches and run to convergence; the tabulated optimizer comparisons use separate sets of 10 restarts.
For $q_{\mathrm{ref}}\neq0$, we define the \textbf{relative quantum advantage} as
\begin{equation*}
    A\coloneqq \max\left\{0, \frac{q_{\mathrm{ref}}-\omega_c}{\left|q_{\mathrm{ref}}\right|} \right\}.
\end{equation*}
An instance believed to have no quantum advantage ($q_{\mathrm{ref}}=\omega_c$) is assigned $A=0$.
We use a \textbf{nominal tolerance} of $1\%$.
Now, the quantum advantage is the quantity of primary
interest for a nonlocal game, and a success criterion should require the optimizer to resolve this
advantage. However, when the nominal 1\% tolerance is comparable to the advantage, a restart that
merely attains the classical value could count as successful. We therefore instead use the
\textbf{effective tolerance}
\begin{equation*}
    \tau_{\mathrm{eff}}=
    \begin{cases}
        {1\%,} & {A\geq2\%,}\\
        0.1\%, & 0.1\%<A<2\%,\\
        {\text{undefined},} & {0\%<A\leq0.1\%,}\\
        {1\%} & {A=0\%,}
    \end{cases}
\end{equation*}
For $0.1\%<A<2\%$, $\tau_{\mathrm{eff}}=0.1\%<A$, so a restart attaining only the classical value does not meet the success criterion. For $0\%<A\leq0.1\%$, the advantage is too small and we do not use these games for benchmarking. We do benchmark unsolved games with $A=0\%$, for which we use the nominal tolerance. For $A\geq2\%$, the nominal tolerance of $1\%$ is sufficient to resolve the quantum advantage.

Let $v(t)$ denote the best feasible value obtained by a restart up to elapsed time $t$, and let $T$ denote the per-restart time limit. For $q_{\mathrm{ref}}\neq0$, a restart is \textbf{successful} if there exists a time $0\leq t<T$ such that the best feasible value reaches the reference from below within the effective tolerance,
\begin{equation*}
    v(t)\geq q_{\mathrm{ref}}-\tau_{\mathrm{eff}}\left|q_{\mathrm{ref}}\right|.
\end{equation*}
Values above $q_{\mathrm{ref}}$ are counted as successful.
If a restart attains the classical value, $v(t)=\omega_c$, then its relative error equals $A$. %

For each method--instance pair, we report the mean first-hit time among successful restarts. A restart that converges or reaches the time limit outside the effective tolerance window is counted as a failure and contributes no time value.
For solved families,~\Cref{tab:eval-known} additionally reports the analytic values $\omega_c$ and $\omega_q$, whereas for unsolved families,~\Cref{tab:eval-unknown} reports $\omega_c$ and the independently computed lower bound $q_{\downarrow}$.

Lastly, for solved families, we want to compare the largest instances that each optimizer can handle.
We define an instance as \textbf{frontier-eligible} whenever its effective tolerance is defined.
Since we only benchmark solved games with $A>0$, frontier-eligibility implies $A>\tau_{\mathrm{eff}}$. Instances with $0\%<A\leq0.1\%$ have no effective tolerance and are excluded from the frontier comparison.
This condition places the classical
value outside the success region, ensuring that the optimizer is doing more than just solving for
the classical value when a quantum advantage actually exists.
The \textbf{frontier-scale problem} is the largest frontier-eligible instance for which at least one of 10 restarts is successful. By computing the size of the frontier-scale problem for different optimizers, we can quantitatively compare their scalability. This frontier-scale problem measures reachability under the stated restart protocol and is not a guarantee of stable successful optimization of games of that size.

These tolerances and the success criterion are conventions; we do not expect the optimizer comparisons below to depend on the particular choices.

\paragraph{Computational resources.}
All comparisons reported here are performed on a single machine. CPU-based BGD and the see-saw method each use one pinned core of an AMD EPYC 7543. GPU-based BGD and vectorized enumeration of deterministic strategies each use one pinned host CPU core and one NVIDIA A100 with $40\,\mathrm{GB}$ of device memory. Each BGD and see-saw task is limited to $64\,\mathrm{GB}$ of host memory, with swap disabled.

\paragraph{BGD.} BGD is implemented in PyTorch in double precision, using automatic differentiation for the gradients and the Adam-based preconditioner and line search described in~\Cref{app:bgd-eng}. Unless stated otherwise, the shared state and local measurement bases are initialized independently for each restart. Furthermore, by default we use rank-one projective measurements without stochastic post-processing, with each local Hilbert-space dimension equal to the corresponding number of outputs.

For the MUB and jamming families, the shared state is fixed to the maximally entangled state of the prescribed local dimension. For the jamming family, this dimension equals the number of outputs, and BGD optimizes the local projective measurements. For MUB, both local Hilbert spaces have dimension~$d$, the dimension in which the existence of the required MUBs is considered. Because Alice has three outputs, when $d\neq3$, BGD uses stochastic post-processing to map the $d$ outcomes of her projective measurement to the three output labels.

We use Adam hyperparameters $(\beta_1,\beta_2)=(0.9,0.99)$ and $\epsilon=10^{-14}$. The base learning rate is $0.1$ for the solved families and CGLMP and $1.0$ for the remaining non-MUB families. The learning rates for the MUB family are listed in~\Cref{subsec:results-mub}. The learning rates are chosen according to the scale of the utility function and do not require prior knowledge of $\omega_q$.

\paragraph{See-saw.} The see-saw method uses the Julia \texttt{Ket.jl} implementation~\cite{KetJl}. Binary-output families take the solver-independent eigenvalue path; the remaining families solve each POVM subproblem as a semidefinite program with the single-threaded MOSEK interior-point solver~\cite{mosek}.

\paragraph{NPA.} Our initial NPA upper bound is obtained from the augmented first-level relaxation described in~\Cref{app:npa-eng}, whose two partition-supported monomial sets represent every correlator in the objective. When the time budget permits, we tighten the bound by increasing the level~$k$ in the same augmented hierarchy. As shown in \Cref{app:npa-eng}, this hierarchy coincides with the standard hierarchy~\cite{navascues2008convergent} when $k\geq\lceil\Nparty/2\rceil$. Each semidefinite program is solved with MOSEK~\cite{mosek} using 16 CPU threads, up to $384\,\mathrm{GiB}$ of host memory, and a $300\,\mathrm{s}$ time limit. The level column in \Cref{tab:eval-unknown} gives the level of the displayed NPA bound; when bounds are obtained at multiple levels, we report the highest such level. Agreement of this upper bound with the best feasible lower bound, within the applicable tolerance, determines $\omega_q$ to that precision.

\paragraph{Classical.} Classical references are obtained from analytic expressions where available, and otherwise by either vectorized GPU enumeration of deterministic strategies or the MILP formulation of \Cref{subsec:milp}. A value is reported as exact only when the enumeration completes or the mixed-integer solver certifies optimality; otherwise, the best feasible value is reported as a classical lower bound and marked with ``$\geq$''. For the MILP optimizer, we use Gurobi~\cite{gurobi} with 8 CPU threads. We report classical-reference computation separately from the optimizer runtimes.

\subsection{Solved families of Bell inequalities}\label{subsec:results-known}
We first consider solved families of Bell inequalities for which we have exact expressions for the classical and quantum values in the general case. We list the families and give their definitions.
Some families are presented in alternative forms of Bell
expressions; their relation to the nonlocal-game formulation
in~\Cref{section:nolocalgame} is explained
in~\Cref{app:expressions}.
We also denote the game shape $(\Nparty, \Nin, \Nout)$ of each family in the paragraph heading. We sometimes use other variables to describe the game shape, depending on the convention used in the literature.

\paragraph{Bell inequalities tailored to maximally entangled states $(2,m,d)$.}
We first consider a family where we can tune the number of inputs and outputs.
   In~\cite{Salavrakos2017Bell}, a class of Bell inequalities tailored to maximally entangled states was proposed.
   Let $m,d \geq 2$ be integers. In terms of conditional probabilities, the Bell expression is
    \begin{equation*}
    \begin{aligned}
        I_{d,m}:=\sum_{k=0}^{d-1} \alpha_k\Big\{&p(a_1=a_2+k|m,m) + p(a_1+k+1=a_2|1,m) \nonumber\\
        & +\sum_{i=1}^{m-1}\Big[p(a_1=a_2+k|i,i) + p(a_1+k=a_2|i+1,i)\Big] \Big\},    
    \end{aligned}
    \end{equation*}
    where equalities involving outputs are understood modulo $d$
    \begin{equation*}
        p(a_1=a_2+k|s_1,s_2):=\sum_{j=0}^{d-1} p((j+k)\mathrm{mod}~d,j|s_1,s_2)
    \end{equation*}
    and
    \begin{equation*}
        \alpha_k \coloneqq \frac{1}{2d}\tan(\frac{\pi}{2m}) \left\{ \cot\left[\frac{\pi}{d}\left( k + \frac{1}{2m}\right)\right] -  \cot\left[\frac{\pi}{d}\left( \lfloor d/2 \rfloor  + \frac{1}{2m}\right)\right]\right\}.
    \end{equation*}

    The classical value of this expression is given by
    \begin{equation*}
    \begin{aligned}
        &\omega_c = \nonumber\\
        & \frac{1}{2d} \tan\frac{\pi}{2m} \left[ (2m-1) \cot \frac{\pi}{2m d} +\cot[\frac{\pi}{d}(d-1+\frac{1}{2m})]\right]  -\frac{m}{d}   \tan \frac{\pi}{2m} \cot \left[\frac \pi d \left(\lfloor d/2 \rfloor +\frac{1}{2m}\right)\right].
    \end{aligned}
    \end{equation*}
    The quantum value is given by
    \begin{equation*}
        \omega_q = \frac{m}{d} \left[ d - \tan \frac{\pi}{2m} \cot\left[\frac \pi d \left(\lfloor d/2\rfloor +\frac{1}{2m}\right)\right] \right].
    \end{equation*}
    As the name of the game suggests, the quantum value is achieved on the maximally entangled state
    \begin{equation*}
        \vert \Phi_d^+\rangle = \frac{1}{\sqrt{d}}\sum_{i=0}^{d-1} \vert ii\rangle.
    \end{equation*}
    For conciseness, we will refer to this family as ``MaxEnt.''

\paragraph{The MABK inequalities $(n,2,2)$.}\label{section:mabk}
We next consider a family where we can tune the number of parties.
Let $n \geq 2$ be an integer.
We consider here the MABK inequalities~\cite{Mermin1990a,Ardehali1992,BelinskyKlyshko1993a}, which are defined recursively for $n\ge 3$ as
\begin{equation*}
    \mathcal{M}_n = \frac{1}{2} \left[ \mathcal{M}_{n-1} \cdot (A^{(n)}_0 + A^{(n)}_1) + \mathcal{M}_{n-1}^\prime \cdot (A^{(n)}_0 - A^{(n)}_1)  \right] \le 1.
\end{equation*}
Here, $\mathcal{M}_{n-1}^\prime$ is obtained by interchanging $A_0^{(k)}$ and $A_1^{(k)}$ in $\mathcal{M}_{n-1}$ for all $k$,
and
\begin{equation*}
    \mathcal{M}_2 \coloneqq \frac{1}{2}\left(\langle A_0^{(1)}A_0^{(2)}\rangle+\langle A_0^{(1)}A_1^{(2)}\rangle + \langle A_1^{(1)}A_0^{(2)}\rangle - \langle A_1^{(1)}A_1^{(2)}\rangle \right)
\end{equation*}
is the CHSH expression normalized to have classical bound one. The operation $\cdot$ ``appends'' a symbol to a correlator. For example,
\begin{equation*}
    \langle  A_1^{(1)} A^{(2)}_1\rangle \cdot A_0^{(3)}   = \langle A_1^{(1)} A_1^{(2)} A_0^{(3)}\rangle.
\end{equation*}
In particular, $\mathcal{M}_3 \leq 1$ is the GHZ inequality.

The classical value of the MABK inequalities is given by
\begin{equation*}
    \omega_c=1
\end{equation*}
while the quantum value is
\begin{equation*}
    \omega_q = 2^{\frac{n-1}{2}}.
\end{equation*}
Thus the ratio $\omega_q/\omega_c$ grows exponentially
with the number of parties.

\paragraph{The Mermin inequalities $(n,2,2)$.}
The Mermin inequalities~\cite{Mermin1990a} also form a family indexed by the number of parties $n\geq3$. We use their winning-probability form. Each party receives a binary input $s_i\in\{0,1\}$ and returns a binary output $a_i\in\{0,1\}$. The input string $\mathbf{s}=(s_1,\ldots,s_n)$ is drawn uniformly from the strings of even Hamming weight $\abs{\mathbf{s}}\coloneqq\sum_{i=1}^n s_i$. The parties win when their output parity equals half this weight modulo $2$:
\begin{equation*}
    \sum_{i=1}^n a_i \equiv \frac{\abs{\mathbf{s}}}{2} \pmod{2}.
\end{equation*}
The optimal classical winning probability is
\begin{equation*}
    \omega_c = \frac{1}{2} + 2^{-\lceil n/2\rceil}.
\end{equation*}
By contrast, a strategy based on the $n$-party GHZ state wins with certainty, so $\omega_q=1$. The quantum--classical gap therefore approaches $1/2$ as $n$ increases.

\paragraph{The chained inequality $(2,m,2)$. }
Let $m \geq 2$ be an integer.
The chained Bell inequality~\cite{Pearle1970,Braunstein1990Wringing}
is another generalization of the CHSH inequality with $m$ inputs given by
\begin{equation*}
    \sum_{i=1}^m \left( \langle A_iB_i\rangle + \langle A_{i+1}B_{i}\rangle \right) \le 2m-2.
\end{equation*}
The correlator term with $A_{m+1}$ is defined to be $-\langle A_1 B_m \rangle$. The CHSH inequality corresponds to $m=2$ with an appropriate permutation of inputs.

The classical value, as seen above, is
\begin{equation*}
    \omega_c = 2m-2.
\end{equation*}
The quantum value is
\begin{equation*}
    \omega_q = 2m\cos \frac{\pi}{2m},
\end{equation*}
which is achieved with the Bell state $\frac{1}{\sqrt 2}(\vert 00\rangle +\vert 11\rangle)$.

\paragraph{Numerical results}
\Cref{tab:eval-known} compares the frontier-scale instances reached by CPU-based BGD, GPU-based BGD, and the see-saw method. We scale the MaxEnt family along $\Nin=\Nout$, at fixed $\Nin=2$, and at fixed $\Nout=2$; MABK and Mermin are scaled in $\Nparty$, and the chained family in $\Nin$.

\begin{table}[b!]
\centering
\caption{Frontier-scale problems for solved families. For each method and scaling axis, we report the largest frontier-eligible instance reached by at least one restart under the protocol of \Cref{subsec:eval-setup}. The asterisk $^*$ marks a problem as frontier-scale because of the frontier-eligibility condition and not because of the time limit. ``succ.'' is the fraction of successful restarts, and ``time'' is their mean time to tolerance, in seconds.}
\label{tab:eval-known}
\small
\setlength{\tabcolsep}{4pt}
\begin{tabular}{l c c c c c c}
method/device & param & $(\Nparty,\Nin,\Nout)$ & $\omega_c$ & $\omega_q$ & succ. & time \\
\hline
\multicolumn{7}{l}{MaxEnt, $\Nin=\Nout=d$} \\
see-saw/CPU & $d=12$ & $(2,12,12)$ & 11.5457 & 12.0014 & 50\% & 254.28 \\
BGD/CPU & $d=42$ & $(2,42,42)$ & 41.5134 & 42.0000 & 70\% & 269.53 \\
BGD/GPU & $d=140$ & $(2,140,140)$ & 139.504 & 140.000 & 10\% & 289.59 \\
\hline
\multicolumn{7}{l}{MaxEnt, $\Nin=2$} \\
see-saw/CPU & $\Nout=23$ & $(2,2,23)$ & 1.69468 & 1.99703 & 30\% & 249.20 \\
BGD/CPU & $\Nout=292$ & $(2,2,292)$ & 1.69767 & 2.00002 & 10\% & 274.67 \\
BGD/GPU & $\Nout=1407$ & $(2,2,1407)$ & 1.69765 & 2.00000 & 10\% & 284.35 \\
\hline
\multicolumn{7}{l}{MaxEnt, $\Nout=2$} \\
see-saw/CPU & $\Nin=78$ & $(2,78,2)$ & 77.5157 & 78.0079 & 10\% & 1.28 \\
BGD/CPU & $\Nin=498^{*}$ & $(2,498,2)$ & 497.502 & 498.001 & 100\% & 19.67 \\
BGD/GPU & $\Nin=498^{*}$ & $(2,498,2)$ & 497.502 & 498.001 & 100\% & 9.61 \\
\hline
\multicolumn{7}{l}{MABK} \\
see-saw/CPU & $\Nparty=8$ & $(8,2,2)$ & 1.00000 & 11.3137 & 90\% & 57.01 \\
BGD/CPU & $\Nparty=11$ & $(11,2,2)$ & 1.00000 & 32.0000 & 100\% & 98.58 \\
BGD/GPU & $\Nparty=13$ & $(13,2,2)$ & 1.00000 & 64.0000 & 100\% & 80.89 \\
\hline
\multicolumn{7}{l}{Mermin} \\
see-saw/CPU & $\Nparty=8$ & $(8,2,2)$ & 0.56250 & 1.00000 & 90\% & 77.78 \\
BGD/CPU & $\Nparty=11$ & $(11,2,2)$ & 0.51562 & 1.00000 & 100\% & 90.53 \\
BGD/GPU & $\Nparty=13$ & $(13,2,2)$ & 0.50781 & 1.00000 & 90\% & 79.49 \\
\hline
\multicolumn{7}{l}{Chained} \\
see-saw/CPU & $\Nin=78$ & $(2,78,2)$ & 154.000 & 155.968 & 10\% & 1.30 \\
BGD/CPU & $\Nin=998^{*}$ & $(2,998,2)$ & 1994.000 & 1995.998 & 100\% & 211.21 \\
BGD/GPU & $\Nin=998^{*}$ & $(2,998,2)$ & 1994.000 & 1995.998 & 100\% & 26.89 \\
\hline
\end{tabular}
\end{table}

In the CPU comparison, BGD reaches a larger frontier than our see-saw implementation on every scaling axis. The separation is largest along the input and output axes and smaller on the party axis.
This contrast is consistent with the per-iteration costs derived in \Cref{app:bigo-seesaw,app:cost}. For example, the local dimension for MaxEnt is $d=\Nout$, so at fixed $\Nin$ the see-saw cost scales with $\Nout$ polynomially as $n_\mathrm{out}^{3\Nparty}=n_\mathrm{out}^6$. In comparison, the BGD cost~\Cref{eq:bgd-cost-projective} scales as $n_\mathrm{out}^{\Nparty+1}=n_\mathrm{out}^3$. This is a per-iteration square-root speedup under the stated fixed-input, bipartite, $d=\Nout$ conditions.
If we instead scale $\Nin$, the costs for both optimizers grow polynomially as $n_\mathrm{in}^{\Nparty}$ to leading order, but with different prefactors: the see-saw assembly costs $\Nparty(\Nin\Nout)^{\Nparty}d^{2\Nparty}$ per sweep, against $O(n_\mathrm{in}^{\Nparty}d^{\,\Nparty+1})$ for BGD, a ratio on the order of $\Nparty n_\mathrm{out}^{\Nparty}d^{\,\Nparty-1}$, which equals $16$ for the bipartite binary-output instances with $d=2$. The see-saw method therefore times out on smaller instances compared to BGD.
Finally, increasing $\Nparty$ leads to exponential growth in the computational cost of both methods, even at fixed local dimension. This accounts for the smaller separation between BGD and see-saw on the MABK and Mermin families. For BGD, evaluating the Bell functional and its gradients requires an exact tensor contraction whose intermediate tensors grow rapidly with the number of parties. The $13$-party GPU instances of both families are reached within the prescribed budget in most restarts, whereas none of the $14$-party frontier-validation instances reaches the target accuracy within that budget.

The entries marked with an asterisk $^*$ are limited by frontier eligibility rather than by the time limit. Both CPU- and GPU-based BGD can handle sufficiently large-scale Bell inequalities where this is a problem. The GPU reduces the endpoint runtime by approximately $2.0\times$ for MaxEnt and $7.9\times$ for the chained family. On runtime-limited families, GPU-based BGD extends the frontier beyond the CPU result.

\begin{figure}[b!]
\centering

\includegraphics[width=0.54\linewidth]
    {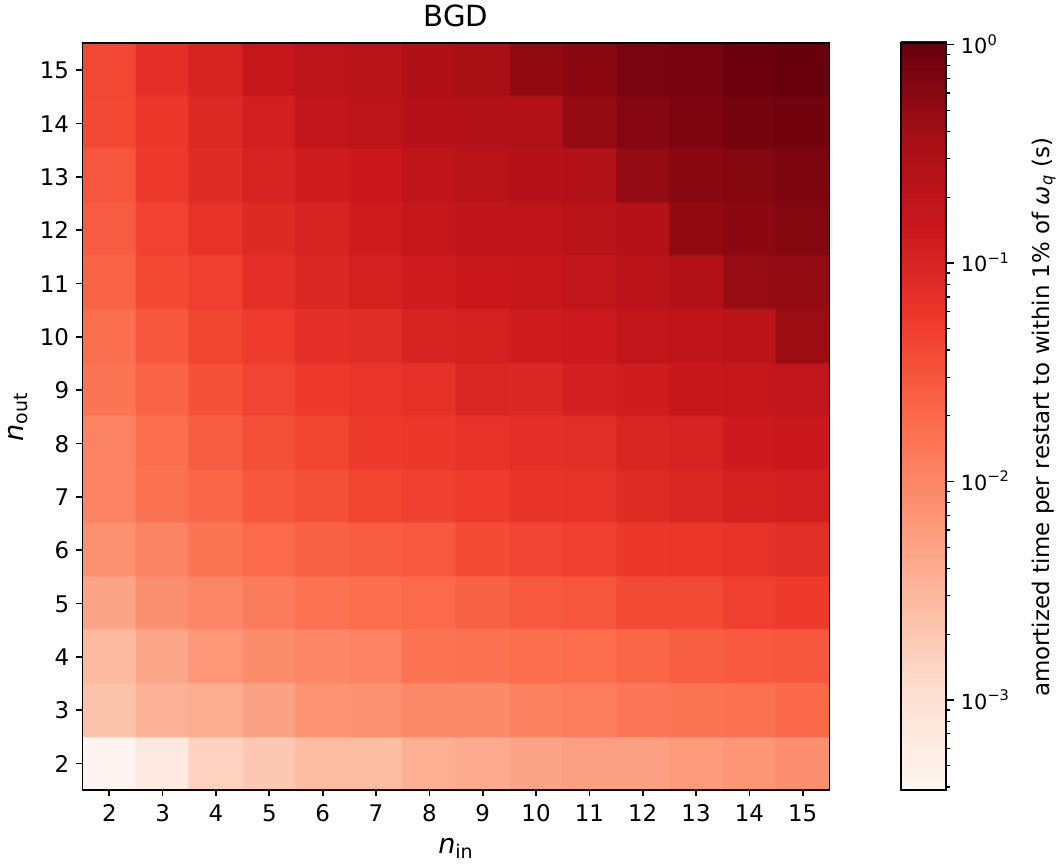}
\caption{Amortized single-core BGD time to reach within
$1\%$ of $\omega_q$ over the MaxEnt grid
$(\Nin,\Nout)\in\{2,\dots,15\}^2$.
Each cell reports the wall-clock time for a batch of
100 independent restarts divided by 100; all cells
reach the target. The color scale is logarithmic.}
\label{fig:scalability-bgd}

\end{figure}

\Cref{fig:scalability-bgd,fig:scalability-seesaw} show how the time to tolerance scales across the MaxEnt parameter grid for each method. The BGD heatmap measures amortized batch throughput, whereas the see-saw heatmap measures successful-restart latency; their numerical times are therefore not directly comparable. Direct comparison is instead provided in \Cref{fig:convergence-solved} for three instances selected so that both profiles remain visible on a linear time axis (in most other instances, BGD completes much earlier than see-saw). BGD advances through many inexpensive gradient steps, whereas see-saw advances through fewer, more expensive alternating updates. Each see-saw trace begins after its first update.

\begin{figure}[t!]
\centering

\includegraphics[width=0.54\linewidth]
    {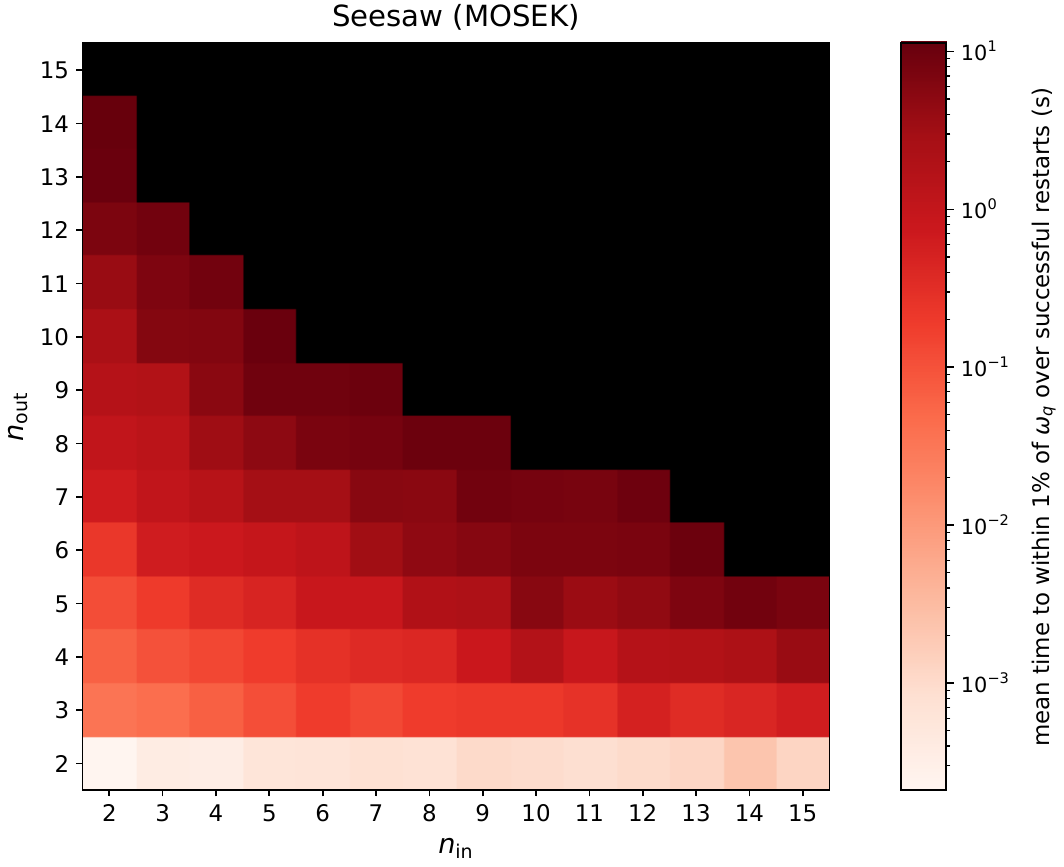}
\caption{Mean single-core see-saw time to reach within
$1\%$ of $\omega_q$ over the MaxEnt grid
$(\Nin,\Nout)\in\{2,\dots,15\}^2$. %
Each cell reports the mean first-hit time over successful
restarts from at most 10 attempts, each limited to
$10\,\mathrm{s}$. Black denotes either an instance for
which no restart reached the target within the time
limit or an instance left untested because a smaller
instance along the corresponding parameter sweep had
already failed. The color scale is logarithmic.}
\label{fig:scalability-seesaw}

\medskip

\includegraphics[width=\linewidth]{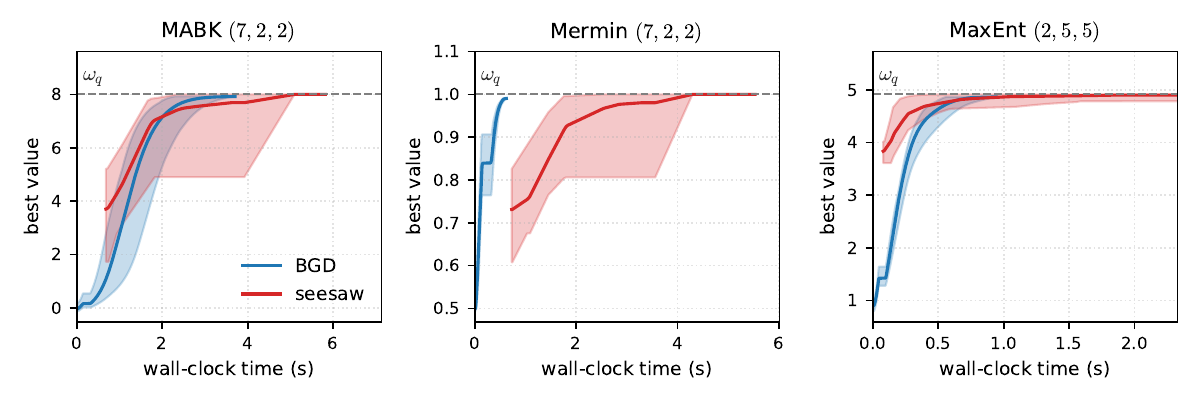}
\caption{Best value against wall-clock time for BGD and see-saw on three solved-family instances under matched, dedicated single-core conditions. Each curve is the mean over 10 independent restarts, the shading spans the minimum and maximum, and the dashed line is the analytic quantum value $\omega_q$. Each panel gives the game shape $(\Nparty,\Nin,\Nout)$.}
\label{fig:convergence-solved}

\end{figure}

\FloatBarrier
\subsection{Unsolved families of Bell inequalities}\label{subsec:results-unknown}
The following families of Bell inequalities lack closed-form expressions for the quantum value in general. For some families, the classical value is also unknown.

\paragraph{The CGLMP inequalities $(2,2,d)$. }
A well-known class of inequalities with a tunable number of outputs is the CGLMP inequalities~\cite{CollinsEtAl2002a}. Let $d \geq 2$ be an integer. For measurement settings $a,b\in\{1,2\}$, define
\begin{equation*}
    p(A_a=B_b+k) \coloneqq \sum_{j=0}^{d-1} p\bigl(A_a=(j+k)\bmod d,\,B_b=j\bigr).
\end{equation*}
Then, the CGLMP Bell expressions are given by
\begin{equation*}
\begin{aligned}
    &I_d = \nonumber\\
    & \sum_{k=0}^{\lceil d/2\rceil -1} \left(1- \frac{2k}{d-1}\right) \big[
    p(A_1=B_1+k) + p(B_1 =A_2+k+1) +p(A_2=B_2+k)  + p(B_2 = A_1+k)\nonumber\\
    &- p(A_1=B_1-k-1) - p(B_1 = A_2-k) - p(A_2=B_2-k-1)-p(B_2 = A_1-k-1)
    \big]
\end{aligned}
\end{equation*}

The classical value is $\omega_c=2$. For the maximally entangled (ME) state
\begin{equation*}
    \ket{\Phi_d^+}=\frac{1}{\sqrt d}\sum_{j=0}^{d-1}\ket{jj}
\end{equation*}
and the standard CGLMP measurements~\cite{CollinsEtAl2002a}, the Bell expression takes the value
\begin{equation*}
    I_d^{\mathrm{ME}} \coloneqq 4d \sum_{k=0}^{\lceil d/2\rceil -1} \left(1- \frac{2k}{d-1}\right) (q_k-q_{-k-1}),
    \qquad
    q_k  \coloneqq \frac{1}{2d^3 \sin^2[\frac \pi d(k+ \frac 1 4 )]},
\end{equation*}
Its asymptotic value is
\begin{equation*}
    I_\infty^{\mathrm{ME}} \coloneqq \lim_{d\to\infty} I_d^{\mathrm{ME}}
    =\frac{32G}{\pi^2}\approx2.96981,
\end{equation*}
where $G$ is Catalan's constant. For each finite $d$, $I_d^{\mathrm{ME}}$ is a feasible quantum value but not, in general, the quantum optimum. Already for $d=3$, the maximally entangled strategy gives $I_3^{\mathrm{ME}}\approx2.87293$, whereas a nonmaximally entangled strategy attains $1+\sqrt{11/3}\approx2.91485$~\cite{AcinEtAl2002}. We therefore treat CGLMP as a family without a known closed-form expression for $\omega_q$.

\paragraph{Load balancing $(2, \Nin, \Nout)$. }
The load-balancing game is a nonlocal game motivated by a real-world coordination problem.
We generalize the nonlocal game defined in~\cite{hasanpour2017quantum} inspired by ad hoc networks, in which two transmitters must choose one of $\Nout$ shared transmission channels without communicating. Each party's input encodes the volume of data it wishes to transmit. Each channel has a capacity. The parties must coordinate their choices so that the capacity of a channel is not exceeded (to avoid congestion) and so that they minimize the overall number of channels used (to conserve resources).
This game isolates the coordination constraint relevant to load balancing and allows us to compare the optimal classical and quantum winning probabilities within the model. Applying our BGD optimizer to this problem demonstrates its use as a numerical tool for real-world multi-agent coordination problems.

Let $\Nin=2k+1$ with $1\leq k\leq20$. The first transmitter has data rate index $s$ and the second transmitter has data rate index $t$, drawn independently and uniformly from $\{-k,-k+1,\ldots,k\}$, corresponding to data rates
\begin{equation*}
    r(s)=1+\delta s,
    \qquad
    r(t)=1+\delta t,
    \qquad
    \delta=0.05.
\end{equation*}
The parties output channel labels $a,b\in\{0,1,\ldots,\Nout-1\}$. Assuming all channels have capacity $C$,
the utility is
\begin{equation}
    V(a,b\vert s,t)=
    \begin{cases}
        \mathbbm 1[a=b], & r(s)+r(t)\leq C,\\
        \mathbbm 1[a\neq b], & r(s)+r(t)>C.
    \end{cases}
    \label{eq:load-balancing-utility}
\end{equation}
Thus matching labels represent sharing a channel at or below capacity, while distinct labels represent splitting traffic above capacity. We will assume $C=2$, reducing the rate condition to $s+t \leq 0$. The classical and quantum values are not known in general for load balancing problems.

The same equality--disequality predicate admits a two-input specialization, although this case is not part of the odd-input family used in our numerical study. For uniformly distributed inputs $x,y\in\{0,1\}$, the parties should output the same channel on $(0,0),(0,1),(1,0)$ and different channels on $(1,1)$. This is the CHSH winning condition with an arbitrary output alphabet. We prove that additional output labels do not change either optimum:
\begin{thm}[Two-input multi-output load balancing]
\label{thm:multi-output-chsh}
For every \(\Nout\in\N\) with \(\Nout\ge2\), the classical and quantum values of the two-input load-balancing game are
\begin{equation}
    \omega_c=\frac34,
    \qquad
    \omega_q=\frac{2+\sqrt2}{4}=\cos^2\frac{\pi}{8}.
    \label{eq:lbchsh-values}
\end{equation}
\end{thm}
\begin{proof}
    We relegate the proof to~\Cref{app:multioutput-chsh}.
\end{proof}

\paragraph{The $(n,k)$ jamming game $(2,\binom{n}{n-k}, n-k)$. }
The jamming game is a two-party coordination task modeling transmission over a contested spectrum~\cite{wehner2026jamming}. An adversary leaves each party a private \emph{safe set} of $d\coloneqq n-k$ usable bands among $n$ frequency bands. Each party receives their safe set as input and outputs an index $a\in\{0,\dots,d-1\}$ selecting one of the safe bands. Each party cannot know the other party's safe set because by construction, they cannot communicate.
Writing $S_x$ for the safe set indexed by input $x$, of which there are $\binom{n}{n-k}$ possibilities, and $S_x[a]$ for its $a$-th band, the parties win when they transmit on the same frequency band:
\begin{equation*}
    V(a,b\vert x,y) =
    \begin{cases}
        1 & S_x[a] =  S_y[b] \\
        0 & \text{otherwise},
    \end{cases}
\end{equation*}
with inputs drawn uniformly. The classical value is given exactly by~\cite{wehner2026jamming}
\begin{equation}
    \omega_c = \frac{1}{\binom{n}{n-k}^2}\sum_{i=0}^{k} \binom{n-1-i}{\,k-i\,}^2 ,
    \label{eq:jamming-classical-value}
\end{equation}
while the quantum value has no general closed form. %
This is another example of a nonlocal game motivated by a real-world application for which our BGD optimizer can serve as a numerical tool.

\paragraph{Random CGLMP network $(\Nparty,2,\Nout)$. }
This family defines multipartite Bell expressions by placing bipartite CGLMP expressions on the edges of a graph. It provides instances with less structure than the individual CGLMP inequalities while retaining interpretability of the game constructed.
Over $\Nparty$ parties with binary inputs and $\Nout$ outputs, we draw a random connected graph whose vertices are the parties, with $\Nparty-1$ edges and maximum degree at most three---a random subcubic tree. We place a copy of the $\Nout$-outcome CGLMP functional on each edge and independently relabel the outputs at its two endpoints. The utility is the average of the edge functionals:
\begin{equation*}
    \mathcal V(\mathbf a\vert \mathbf s) = \frac{1}{\abs{E}}\sum_{e=(i,j)\in E} \mathcal V_{\Nout}^{\mathrm{CGLMP}}\big(\pi_{i,e}(a_i),\pi_{j,e}(a_j)\vert s_i,s_j\big),
\end{equation*}
where $\mathcal V_{\Nout}^{\mathrm{CGLMP}}$ is the coefficient function of the $\Nout$-outcome CGLMP expression defined in~\cite{CollinsEtAl2002a}, and $\pi_{i,e}$ is the output permutation assigned to endpoint $i$ of edge $e$. The global input tuple $\mathbf s\in\{0,1\}^{\Nparty}$ is sampled uniformly.
Note the graph specifies the Bell functional only: the optimization permits an arbitrary state shared by all parties and imposes no independent-source constraint. The degree bound limits the number of pairwise CGLMP terms incident on any party, thereby reducing the competition among incompatible correlations on adjacent edges. No closed-form expression for the quantum value of this random family is known.%

\paragraph{Uniformly random game $(\Nparty,\Nin,\Nout)$. }
To probe the optimizers on Bell inequalities with no prescribed structure whatsoever, we sample a Bell functional
\begin{equation}
    \mathcal B_R[p]
    \coloneqq
    \sum_{\mathbf s\in[\Nin]^{\Nparty}}
    \sum_{\mathbf a\in[\Nout]^{\Nparty}}
    R_{\mathbf s,\mathbf a}\,p(\mathbf a\vert\mathbf s),
    \qquad
    R_{\mathbf s,\mathbf a}\overset{\mathrm{i.i.d.}}{\sim}\operatorname{Unif}[-1,1].
    \label{eq:uniform-random-game}
\end{equation}
The implementation and our numerical experiments use the unnormalized functional $\mathcal B_R$; dividing it by $\Nin^{\Nparty}$ gives the expected utility for uniformly distributed inputs and payoff $R_{\mathbf s,\mathbf a}$. The tabulated instances have $(\Nin,\Nout)=(2,2)$. These random instances have no known closed-form classical or quantum value.

\paragraph{Numerical results}
\Cref{tab:eval-unknown} reports two parameter choices for each family. The smaller instance supports a matched optimizer comparison between the see-saw method and the BGD optimizer, while the larger instance probes a regime in which see-saw fails while GPU BGD still returns a feasible lower bound. 

Because the exact value of $\omega_q$ is unknown for these families, we bracket it using complementary bounds. BGD and the see-saw method provide lower bounds through explicit feasible strategies, while the NPA hierarchy provides an upper bound whenever it can be completed within the resource limit.
A strict inequality $q_{\downarrow}>\omega_c$ certifies a Bell violation whenever the displayed classical value is exact. Agreement between $q_{\downarrow}$ and the NPA upper bound determines $\omega_q$ to the displayed precision.

The classical values for the jamming and CGLMP families are known. GPU enumeration can compute the classical values of the random CGLMP network and uniformly random families, including the nine-party binary-input, binary-output uniformly random game. For the $k=6$ load balancing instance, we found that exhaustive enumeration is infeasible, while our MILP optimizer can certify optimality. \emph{This shows that our MILP solver can possibly handle larger games compared to brute force search}. For the larger $k=20$ load balancing instance, neither GPU enumeration nor the MILP solver could certify optimality.
We instead compute a lower bound on $\omega_c$ in the same way as how we compute $q_\downarrow$: using our BGD optimizer with 256 random initializations, but where the quantum dimension is set to $d=1$ (With the stochastic map, this results in a classical strategy.). Note only for this one instance we use a lower bound on the classical value as a numerical proxy for the classical value.

\begin{table}[t!]
\centering
\caption{Unsolved-family instances at the selected parameter points. ``RandCGLMP'' refers to the random CGLMP network family.
``succ.'' and ``time'' denote the successful-restart fraction and mean time to tolerance, respectively; ``lvl'' is the completed NPA level. A dash ``---'' denotes an unavailable result or an NPA bound not obtained within the resource limit, and ``$\geq$'' marks a classical lower bound.}
\label{tab:eval-unknown}
\small
\renewcommand{\arraystretch}{1.15}
\setlength{\tabcolsep}{3pt}
\begin{tabular}{l c c c c c c c c c c c}
 & & & & \multicolumn{2}{c}{BGD/CPU} & \multicolumn{2}{c}{BGD/GPU} & \multicolumn{2}{c}{see-saw} & \multicolumn{2}{c}{NPA} \\
\cmidrule(lr{2pt}){5-6}
\cmidrule(lr{2pt}){7-8}
\cmidrule(lr{2pt}){9-10}
\cmidrule(lr{2pt}){11-12}
param & $(\Nparty,\Nin,\Nout)$ & $\omega_c$ & $q_{\downarrow}$ & succ. & time & succ. & time & succ. & time & lvl & bound \\
\hline
\multicolumn{12}{l}{CGLMP} \\
\quad $d{=}18$  & $(2,2,18)$  & 2.00000 & 3.23346 & 50\%  & 0.711 & 50\%  & 1.488 & 100\% & 185.155 & 1 & 7.62176 \\
\quad $d{=}179$ & $(2,2,179)$ & 2.00000 & 3.43012 & 100\% & 211.084 & 100\% & 4.414 & 0\%    & --- & --- & --- \\
\hline
\multicolumn{12}{l}{Load balancing, $\Nin{=}\Nout$} \\
\quad $k{=}6$ & $(2,13,13)$ & 0.82249 & 0.84064 & 60\%  & 1.345 & 60\%  & 1.638 & 100\% & 101.671 & --- & --- \\
\quad $k{=}20$ & $(2,41,41)$ & $\geq 0.82927$ & 0.84715 & 100\% & 149.139  & 100\% & 2.943 & 0\%    & ---  & --- & --- \\
\hline
\multicolumn{12}{l}{Jamming, $\Nout{=}4$} \\
\quad $n{=}11, d=4$ & $(2,330,4)$  & 0.24182 & 0.25389 & 100\% & 12.947 & 100\% & 3.464 & 60\% & 51.354 & --- & --- \\
\quad $n{=}16, d=4$ & $(2,1820,4)$ & 0.15797 & 0.16748 & 0\% & ---  & 100\% & 12.153 & 0\% & ---  & --- & --- \\
\hline
\multicolumn{12}{l}{RandCGLMP, $\Nparty{=}4$} \\
\quad $\Nout{=}4$ & $(4,2,4)$ & 0.50000 & 0.54602 & 90\% & 1.088 & 90\% & 3.156 & 10\% & 28.685 & 2 & 0.55667 \\
\quad $\Nout{=}6$ & $(4,2,6)$ & 0.50000 & 0.58762 & 50\% & 1.979 & 50\% & 4.038 & 0\% & --- & 1 & 0.99091 \\
\hline
\multicolumn{12}{l}{Uniformly random} \\
\quad $\Nparty{=}7$ & $(7,2,2)$ & 27.5956 & 27.5956 & 0\%   & ---  & 20\% & 3.369 & 10\% & 7.517 & 3 & 27.5956 \\
\quad $\Nparty{=}9$ & $(9,2,2)$ & 63.9107 & 63.9107 & 10\% & 8.033 & 20\% & 3.916 & 0\%   & ---  & --- & --- \\
\hline
\end{tabular}
\end{table}

For CGLMP, all methods reach the reference at $d=18$, whereas at $d=179$ both BGD variants attain unit success and the see-saw method does not achieve the effective tolerance within the resource limit.

For load balancing with capacity $C=2.0$, see-saw succeeds on the smaller instance but not on the larger one; meanwhile, both BGD variants succeed at both sizes. \emph{Such computed values can directly
inform how load balancing can be performed using quantum nonlocality in real-world systems}. 

For the jamming family with $\Nout=4$, both BGD variants succeed at $n=11$, while at $n=16$ GPU BGD succeeds in all ten restarts and CPU BGD succeeds in none. The see-saw success rate is $60\%$ at $n=11$ and $0\%$ at $n=16$. At $n=16$, $q_{\downarrow}=0.16748$ exceeds the exact classical value $\omega_c=0.15797$ from~\Cref{eq:jamming-classical-value} by approximately 6.0\%, certifying a Bell violation. 

For the selected random CGLMP networks, $q_{\downarrow}$ exceeds
the exact classical value $0.50000$ for both the $(4,2,4)$ star
instance and the $(4,2,6)$ path instance. Note that the overall utility of the network averages the CGLMP functionals on each edge, so its normalization differs from the standalone CGLMP value. BGD reaches the reference target in both instances, while
see-saw succeeds only for the smaller $(4,2,4)$ star instance under
the stated protocol. The sampled networks and their output relabelings are specified in~\Cref{app:cglmpgraph-instances}.

For the uniformly random family, the $q_\downarrow$ values equal the exact classical values to the displayed precision. At seven parties, the level-3 NPA upper bound agrees with the classical value, certifying $\omega_q=\omega_c$ to that precision. At nine parties, only BGD reaches the reference. Because the two party counts correspond to different random Bell functionals, these restart rates should not be interpreted as a monotonic trend with the number of parties. The higher CPU success rate at nine parties can be attributed to the randomness of the optimization.

\Cref{fig:convergence-unsolved} compares four representative unsolved-family instances selected so that both methods remain visible on a linear time axis. The CGLMP and load-balancing panels use matched dedicated restart sets; the jamming panel uses a separate restart set at $(n,k)=(6,2)$, with game shape $(2,15,4)$.

\begin{figure}[htbp!]
\centering
\includegraphics[width=0.93\linewidth]{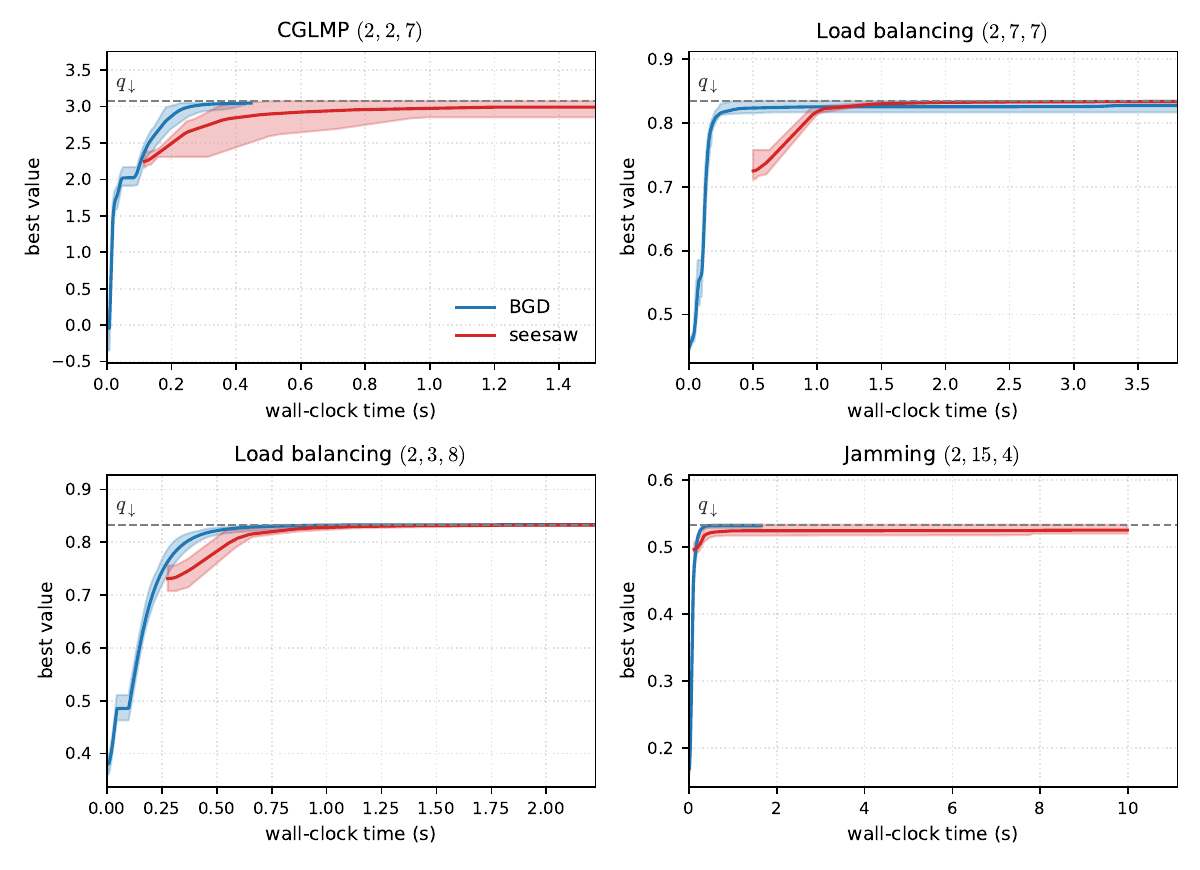}
\caption{Best value against wall-clock time for BGD and see-saw on four unsolved-family instances. Each curve is the mean over 10 independent
restarts, the shading spans the minimum and maximum, and the dashed line is the feasible numerical reference $q_\downarrow$. Each panel gives the game shape $(\Nparty,\Nin,\Nout)$.}
\label{fig:convergence-unsolved}
\end{figure}

\paragraph{Jamming family: comparison with seed-frame optimization.}
We additionally compare BGD with the specialized seed-frame optimizer of~\cite{wehner2026jamming}. Just as in their study, we fix the shared state to the maximally entangled state. The seed-frame optimizer imposes structure on the measurements: Bob's basis is the complex conjugate of Alice's for each safe set. In the seed-frame ansatz, $n$ seed vectors determine the measurement bases through L\"owdin orthogonalization. BGD in contrast varies Alice's and Bob's projective bases independently, thereby searching a broader class of measurement strategies than the seed-frame ansatz. We compare with the seed-frame optimizer because it always attains the best feasible value compared to other proposed optimizers to the displayed precision for the game parameters chosen, as shown in Table VI of~\cite{wehner2026jamming}.

\begin{table}[htbp!]
\centering
\caption{BGD comparison with the optimized seed-frame values reported in Table~I of~\cite{wehner2026jamming}. The values are winning probabilities. All optimizers fix the shared quantum state to be the maximally entangled state.}
\label{tab:jamming-wehner-comparison}
\small
\renewcommand{\arraystretch}{1.05}
\setlength{\tabcolsep}{4pt}
\begin{tabular}{c c c c c c c}
\toprule
$n$ & $d$ & $k$ & $\omega_c$ & Opt. & BGD/CPU & BGD/GPU \\
\midrule
3 & 2 & 1 & 0.5556 & 0.5833 & 0.5833 & 0.5833 \\
4 & 2 & 2 & 0.3889 & 0.4167 & 0.4167 & 0.4167 \\
5 & 2 & 3 & 0.3000 & 0.3230 & 0.3230 & 0.3230 \\
6 & 2 & 4 & 0.2444 & 0.2644 & 0.2644 & 0.2644 \\
4 & 3 & 1 & 0.6250 & 0.6451 & 0.6451 & 0.6451 \\
5 & 3 & 2 & 0.4600 & 0.4808 & 0.4809 & 0.4809 \\
6 & 3 & 3 & 0.3650 & 0.3852 & 0.3852 & 0.3852 \\
5 & 4 & 1 & 0.6800 & 0.6909 & 0.6909 & 0.6909 \\
6 & 4 & 2 & 0.5200 & 0.5333 & 0.5333 & 0.5333 \\
6 & 5 & 1 & 0.7222 & 0.7260 & 0.7260 & 0.7260 \\
7 & 6 & 1 & 0.7551 & 0.7537 & 0.7551 & 0.7551 \\
\bottomrule
\end{tabular}
\end{table}

\Cref{tab:jamming-wehner-comparison} compares the feasible values of the two optimizers.
In particular, we use Table~I of~\cite{wehner2026jamming} for the numerically optimized seed-frame values, denoted there by ``Opt.'', for eleven sets of parameters.
At the displayed precision, nine rows match exactly. The $(n,d,k)=(5,3,2)$ row is $0.4808$ versus $0.4809$. For $(n,d,k)=(7,6,1)$, BGD raises the feasible value from $0.7537$ to $0.7551$, reaching $\omega_c$ to numerical precision.
Hence our general purpose BGD optimizer performs at least as well as the specialized optimizer on these instances.
Note that because the methods search different measurement classes, their attained feasible values are determined by the numerical optimizer quality and the choice of the measurement bases ansatz.

The corresponding runtime data are listed separately in~\Cref{tab:jamming-wehner-runtime}.
Because~\cite{wehner2026jamming} does not report runtimes, we run the public seed-frame implementation accompanying that work~\cite{quantumjamming} on our specified CPU and measure its wall-clock time under the same ten-restart protocol used for BGD. The resulting runtimes are directly comparable with the CPU BGD times in~\Cref{tab:jamming-wehner-runtime}.
The seed-frame restarts show steep runtime growth as the number of frequency bands and safe sets increase, whereas BGD restarts remain in the low-single-digit-second range on these instances. CPU BGD is faster than GPU BGD throughout this set, consistent with the GPU launch overhead at these small problem sizes. Both BGD implementations are substantially faster than the seed-frame restarts once the latter becomes expensive.

\begin{table}[htbp!]
\centering
\caption{Mean wall-clock time, in seconds, over ten valid restarts.}
\label{tab:jamming-wehner-runtime}
\small
\renewcommand{\arraystretch}{1.05}
\setlength{\tabcolsep}{4pt}
\begin{tabular}{c c c c c c}
\toprule
$n$ & $d$ & $k$ & seed/CPU & BGD/CPU & BGD/GPU \\
\midrule
3 & 2 & 1 & 0.116 & 0.871 & 2.097 \\
4 & 2 & 2 & 0.705 & 0.983 & 2.317 \\
5 & 2 & 3 & 2.213 & 1.016 & 2.300 \\
6 & 2 & 4 & 6.009 & 1.294 & 2.621 \\
4 & 3 & 1 & 0.814 & 1.308 & 2.343 \\
5 & 3 & 2 & 35.436 & 1.623 & 3.702 \\
6 & 3 & 3 & 287.104 & 1.727 & 2.919 \\
5 & 4 & 1 & 3.395 & 1.623 & 3.347 \\
6 & 4 & 2 & 3.303 & 2.517 & 4.117 \\
6 & 5 & 1 & 4.350 & 2.010 & 3.571 \\
7 & 6 & 1 & 15.654 & 2.553 & 5.382 \\
\bottomrule
\end{tabular}
\end{table}

\subsection{Mutually unbiased bases inequalities}\label{subsec:results-mub}
We finally consider the mutually unbiased basis (MUB) inequalities~\cite{Colomer2022threenumerical}.

\paragraph{Mutually unbiased bases inequalities $(2, [n(n-1)d^2/2,n], [3,d])$.}
Let $n,d\geq 2$ be integers. A family of orthonormal bases $\{\mathcal{B}_y\}_{y=0}^{n-1}$ in $\C^d$, with
\begin{equation*}
    \mathcal{B}_y=\{\ket{b^y_0},\ket{b^y_1},\cdots,\ket{b^y_{d-1}}\},
\end{equation*}
is called a family of \textbf{mutually unbiased bases} (MUBs) if for any $y\neq z$,
\begin{equation*}
    \abs{\langle b^y_j \vert b^z_k\rangle}^2=\frac{1}{d},
    \qquad \forall j,k\in\{0,1,\cdots,d-1\}.
\end{equation*}
In other words, if the system is prepared in a basis vector from one basis, then all outcomes of a measurement in a different basis occur with equal probability. A complete set of MUBs contains $d+1$ bases, and such complete sets are known to exist whenever $d$ is a prime power.

Here we introduce a class of bipartite Bell inequalities tailored to MUBs~\cite{Colomer2022threenumerical}. Bob receives an input $t\in\{0,1,\cdots,n-1\}$ and outputs $b\in\{0,1,\cdots,d-1\}$. Alice receives an input
\begin{equation*}
    x=(y,z,j,k),
\end{equation*}
where $0\le y<z\le n-1$ and $j,k\in\{0,1,\cdots,d-1\}$, and outputs
\begin{equation*}
    a\in\{1,2,\perp\},
\end{equation*}
where $\perp$ is some symbol.
The Bell expression can be written as
\begin{equation*}
\begin{aligned}
    \mathcal{S}_{d,n}
    :=&\sum_{0\le y<z\le n-1}\sum_{j,k=0}^{d-1}
    \Bigl[
        p(a=1,b=j|x=(y,z,j,k),t=y)
        -p(a=1,b=k|x=(y,z,j,k),t=z)\nonumber\\
        &-p(a=2,b=j|x=(y,z,j,k),t=y)
        +p(a=2,b=k|x=(y,z,j,k),t=z)\nonumber\\
        &-\gamma\, p_A(a=1|x=(y,z,j,k))
        -\gamma\, p_A(a=2|x=(y,z,j,k))
    \Bigr]
    ,
\end{aligned}
\end{equation*}
where
\begin{equation*}
    \gamma=\frac{1}{2}\sqrt{\frac{d-1}{d}},
\end{equation*}
and $p_A(a|x)$ denotes Alice's marginal probability.

The classical value is
\begin{equation*}
    \omega_c = \frac{n(n-1)}{2} \cdot 2(d-1) \left(1- \frac 1 2 \sqrt{\frac{d-1}{d}}\right).
\end{equation*}
And the quantum value is upper bounded by 
\begin{equation}
\label{eq:mub_qval}
    \omega_q \le \omega_{\mathrm{MUB}}
    \coloneqq \binom{n}{2}\sqrt{d(d-1)}
    = \tfrac12 W_{\mathrm{MUB}}(d,n).
\end{equation}
This bound is attainable in local dimension $d$
if and only if $n$ MUBs exist in that
dimension~\cite{Colomer2022threenumerical}.
The table below contains only parameter pairs for which the required number of MUBs is known to exist; for prime-power dimensions this follows by taking subsets of a complete set~\cite{WoottersFields1989}.

Reference~\cite{Colomer2022threenumerical} studies this
MUB-search problem using see-saw, nonlinear semidefinite
programming, and Monte Carlo methods.
For the strategies considered there, the reported
objective $W$ corresponds to the Bell value
$\mathcal S_{\mathrm{ref}}=W-\tfrac12W_{\mathrm{MUB}}$.
Hence
\begin{equation*}
    1-\frac{\mathcal S_{\mathrm{ref}}}{\omega_{\mathrm{MUB}}}
    =2\left(1-\frac{W}{W_{\mathrm{MUB}}}\right).
\end{equation*}
We therefore multiply the best reported relative gap
by two to obtain the Bell-value gap shown
in~\Cref{tab:eval-mub}.

\paragraph{Numerical results}
The base learning rate for MUB is $0.2$ by default. The calculations use $0.01$ for $(d,n)=(9,10),(11,12)$, $0.6$ for $(10,3)$, and $0.1$ for $(12,4)$. It is known that there are 4 MUBs in $d=12$, 3 MUBs in $d=6$, and 3 MUBs in $d=10$ via the standard tensor product construction~\cite{McNultyWeigert2026}.

\begin{table}[htbp!]
\centering
\caption{Fixed-state BGD values for the MUB inequalities. The BGD-gap column is the relative gap to the analytic optimum $\omega_{\mathrm{MUB}}=\omega_q=\tfrac12 W_{\mathrm{MUB}}$. The final column gives the best reported Bell-value relative gap from~\cite{Colomer2022threenumerical}. ``--'' denotes an unreported parameter pair. A gap of $0.000\%$ implies the quantum value is attained to the displayed precision.}
\label{tab:eval-mub}
\small
\setlength{\tabcolsep}{3.5pt}
\begin{tabular}{c c c c c c c}
\hline
$d$ & $n$ & $\omega_c$ & $\omega_{\mathrm{MUB}}$ & BGD & BGD gap & Ref. best gap \\
\hline
2  & 3  & 3.8787   & 4.2426   & 4.2426   & 0.000\% & 0.000\% \\
3  & 4  & 14.2020  & 14.6969  & 14.6969  & 0.000\% & 0.000\% \\
4  & 5  & 34.0192  & 34.6410  & 34.6410  & 0.000\% & 0.000\% \\
5  & 6  & 66.3344  & 67.0820  & 67.0820  & 0.000\% & 0.000\% \\
6  & 3  & 16.3069  & 16.4317  & 16.4317  & 0.000\% & 0.000\% \\
7  & 8  & 180.4622 & 181.4607 & 181.4546 & 0.003\% & 0.000\% \\
9  & 10 & 380.5887 & 381.8377 & 381.5623 & 0.072\% & --      \\
10 & 3  & 28.3856  & 28.4605  & 28.4602  & 0.001\% & 0.000\% \\
11 & 12 & 690.7147 & 692.2138 & 691.7924 & 0.061\% & --      \\
12 & 4  & 68.8098  & 68.9348 & 68.9288  & 0.009\% & --      \\
\hline
\end{tabular}
\end{table}

Table~\ref{tab:eval-mub} compares the fixed-state BGD values with the quantum value and also the best feasible value obtained in~\cite{Colomer2022threenumerical}.
The see-saw and nonlinear-SDP methods in~\cite{Colomer2022threenumerical} reach the quantum value to their displayed precision at the five parameter pairs $(d,n)=(2,3),(3,4),(4,5),(5,6),(6,3)$. The Monte Carlo method reaches the same precision at those pairs and at $(7,8)$ and $(10,3)$. BGD reaches the same displayed precision at the five pairs shared with the see-saw and nonlinear-SDP calculations. The gap is small but nonzero for $(7,8)$ and $(10,3)$. The largest gaps are $0.072\%$ at $(9,10)$, $0.061\%$ at $(11,12)$, and $0.009\%$ at $(12,4)$. These three parameter pairs are not reported in~\cite{Colomer2022threenumerical}. 
Thus BGD, a general-purpose optimizer, obtains results comparable to those of the specialized optimizers on the reported instances. Note that as a local first-order method, BGD is not designed for computing the quantum value to very high precision.

\section{Discussion}\label{sec:discussion}
In this paper, we presented an optimizer for computing lower bounds on the quantum value for a general nonlocal game based on batched gradient descent (BGD). Our BGD optimizer evaluates the expected utility by a direct tensor contraction, never forming the dense $d^{\Nparty}\times d^{\Nparty}$ Bell operator. This lets it handle problem sizes beyond the standard numerical methods such as see-saw. In our numerical studies it performs remarkably well: it recovers the quantum value of the solved families to high precision and  computes high-quality lower bounds on unsolved families.
From our numerical experiments, we find that across all families, BGD handles games well beyond the reach of a standard off-the-shelf implementation of the see-saw method because for the former, per-iteration cost is dominated by the contraction against the weighted utility tensor rather than the latter's repeated dense eigenproblems. It is general-purpose, able to optimize any game specified by a utility function and an input distribution, and its batched restarts naturally parallelize on a GPU to give a good estimate of the quantum value in mere seconds for instances where the see-saw method requires hundreds of seconds or more.

However, these gains do not remove the exponential dependence on the number of parties $\Nparty$. The weighted utility tensor $W$ has  $(\Nin\Nout)^{\Nparty}$ entries. Its size is therefore exponential in $\Nparty$, but, for fixed $\Nparty$, it is polynomial of degree $\Nparty$ in each of $\Nin$ and $\Nout$. Moreover, although BGD avoids forming the Bell operator, it represents the shared quantum state explicitly by its $d^{\Nparty}$ complex amplitudes, so storing the state and evaluating contractions involving it remain exponential in the number of parties. The isolated measurements in \Cref{fig:cost-scaling} of~\Cref{app:cost} exhibit these predicted dependencies once the contractions become compute-bound. Further gains may therefore come from exploiting sparsity, symmetry, or locality in the game specification, and from replacing the unrestricted state vector with one from a structured variational family. This is perfectly sensible for practical applications where the quantum state is prepared via a shallow quantum circuit and therefore can be compactly represented with a tensor network~\cite{Markov2008Simulating}. Similarly, the quantum measurements are also implemented by a shallow quantum circuit.
The tensor network describing the Bell expression
could then be efficiently contracted using techniques such as that of~\cite{boixo2017simulation,villalonga2019flexible,Huang2021Efficient}. The effectiveness of this approach could depend on the form of the utility function.

The explicit dependence on the weighted utility tensor can also be relaxed. The expected utility is the average of the utility $\mathcal V(\mathbf a|\mathbf s)$ over inputs drawn from the input distribution $\pi$ and outputs drawn from the behavior $p(\mathbf a|\mathbf s)$, so its gradient admits a score-function (REINFORCE) estimate from sampled measurement records $(\mathbf s,\mathbf a)$ and the gradients of their log-probabilities $\log p(\mathbf a|\mathbf s)$, which the state vector provides. Optimizing from such samples queries the utility only at the sampled input--output tuples rather than reading all $(\Nin\Nout)^{\Nparty}$ entries.

Our optimizer substantially expands the space of Bell inequalities that can be numerically analyzed. In particular, it can serve as a numerical evaluation tool for multi-agent coordination problems that can be modeled as nonlocal games~\cite{ding2026quantum}. Real-world coordination problems would involve utility functions and input distributions based on real data. Hence, a nonlocal game that models such problems may have a large number of inputs and outputs. Such games can in principle be efficiently handled by our BGD optimizer. Furthermore, in future work we can generalize our optimizer to handle more complicated scenarios involving real-time communication between parties with time-dependent input and output processes~\cite{ding2025quantum}. For evaluating physical implementations for violating Bell inequalities, future extensions of our optimizer could also take into account noise~\cite{CiscoQuantumSync2025} or photon loss~\cite{PhysRevA.47.R747,Xu2026Violating}.

\paragraph{Acknowledgments} We thank Mengyao Hu for helpful discussions. DD would like to thank God for all of His provisions.

\paragraph{AI Disclosure} Claude-family and GPT-family models were also used to assist with writing and proofreading the manuscript, as well as checking the code used for the numerical experiments. The proof of the two-input load balancing bound in this work was developed with the assistance of GPT-5.6 Sol.

\bibliographystyle{unsrt}
\bibliography{ref}

\addtocontents{toc}{\protect\setcounter{tocdepth}{1}}
\appendix

\section{See-Saw Implementation Details}
 \label[appendix]{app:seesaw-eng}
This appendix details the see-saw method introduced in \Cref{subsec:existing_quantum}, as well as the closed-form POVM update for binary outputs and the dense Hermitian matrix eigensolver used for the state update. The per-sweep computational complexity is also derived.

\paragraph{Modifications to the Ket.jl implementation.}
The implementation used in this work is based on an upstream version of
\texttt{Ket.jl}~\cite{KetJl} that already contains the multipartite see-saw procedure.\footnote{The
modified see-saw code used in the numerical experiments is available at \url{https://github.com/JubilantJerry/quantum-bell-inequality-optim-seesaw}. } We retain the alternating optimization over the
parties' measurements, but modify the binary-output bipartite state update to use the top
eigenvector of the full $d^2\times d^2$ Bell operator rather than a
fixed-basis Schmidt-state update. We also use the convergence condition
\begin{equation*}
    \left|\omega_t-\omega_{t-1}\right|\leq\varepsilon,
\end{equation*}
and allow up to $1000$ sweeps. These changes improve numerical reliability
without changing the underlying see-saw optimization problem.

\subsection{Description of algorithm}

\paragraph{Sweep structure.} Each sweep computes the coefficient matrices $\Gamma_{s_k}^{(k)}$ for one party, optimizes that party's POVMs while holding the rest of the strategy fixed, and repeats this measurement update for every party. Unless the shared state is fixed by the benchmark protocol,
the sweep then optimizes the shared quantum state $\ket{\psi}$. Repeating these sweeps produces a nondecreasing sequence of expected utilities when each subproblem is solved exactly. The final result generally depends on the initialization.

\paragraph{Coefficient matrices by partial trace.} Both the general POVM update and its binary-output specialization reduce the $\Nparty$-party Bell operator to an effective operator on the party being updated. With the state $\rho=\ketbra{\psi}$ and every other party's measurement held fixed, each Bell term contributes to party $k$ the $d\times d$ Hermitian operator $\operatorname{Tr}_{\neq k}\!\big[\rho\,\big(\bigotimes_{j\neq k}M_j\big)\otimes\id_k\big]$, obtained by tracing out all other parties after applying their fixed measurement operators. Summing these contributions with the Bell coefficients gives the matrices $\Gamma^{(k)}_{s_k,a_k}$ in~\Cref{eq:seesaw_povm_sdp}. For general outputs they define the POVM SDP
\begin{equation*}
    \max \sum_{s_k,a_k}\tr\!\big(\Gamma^{(k)}_{s_k,a_k}\, M_{a_k\vert s_k}\big)
    \quad\text{s.t.}\quad
    M_{a_k\vert s_k}\succeq 0,\ \ \sum_{a_k}M_{a_k\vert s_k}=\id_d.
\end{equation*}
For binary outputs, completeness constraints $\sum_{a_k}M_{a_k\vert s_k}=\id_d$ leave one independent measurement operator per input and hence one coefficient matrix $\Gamma^{(k)}_{s_k}$ per input.

\paragraph{Binary-output projection.} When every output is binary, each POVM is determined by one effect $0\preceq E_k(s_k)\preceq\id_d$, with the second element being $\id_d-E_k(s_k)$. The linear objective is maximized by choosing $E_k(s_k)$ to be the projector $P_k(s_k)$ onto the positive eigenspace of $\Gamma^{(k)}_{s_k} \coloneqq \Gamma^{(k)}_{s_k,0} - \Gamma^{(k)}_{s_k,1}$:
\begin{equation}
    P_k(s_k)  \coloneqq \!\!\sum_{i:\,\lambda_i>0}\!\ketbra{u_i},
    \qquad \Gamma^{(k)}_{s_k}=\sum_i \lambda_i\ketbra{u_i},
    \label{eq:seesaw_positive_part}
\end{equation}
Every update is then an eigendecomposition---of a $d\times d$ matrix for a projector, of a $d^{\Nparty}$ one for the state---rather than an SDP, which optimized eigensolvers run far faster in practice.

\paragraph{State eigenproblem.} After a sweep over all parties, maximizing $\bra{\psi}\mathcal B\ket{\psi}$ updates the state to a top eigenvector of the current Bell operator and returns $\lambda_{\max}(\mathcal B)$. This is an eigenproblem for a $D\times D$ Hermitian matrix with $D=d^{\Nparty}$. Our implementation in the Julia \texttt{Ket.jl} library~\cite{KetJl} uses LAPACK's ranged dense Hermitian eigensolver to compute only the largest eigenpair. This update globally solves the state subproblem and therefore cannot decrease the expected utility, up to numerical precision, even when the top eigenspace is degenerate.%

\subsection{Complexity analysis}
\label[appendix]{app:bigo-seesaw}
Assume a game shape with $\Nparty$ parties, $\Nin$ inputs, $\Nout$ outputs, and a common quantum dimension $d$. A see-saw sweep performs $\Nparty$ per-party POVM updates and one state update.

\paragraph{State update.} The state to be optimized is the top eigenpair of the assembled $d^{\Nparty} \times d^{\Nparty}$ Bell operator. First, assembling the Bell operator involves summing $(\Nin\Nout)^{\Nparty}$ terms, with each term being a $d^{\Nparty}\times d^{\Nparty}$ matrix. This requires $O\!\big((\Nin\Nout)^{\Nparty}d^{2\Nparty}\big)$ operations. Next, we perform a dense Hermitian decomposition of the Bell operator, incurring a cost of $O(d^{3\Nparty})$ operations.

\paragraph{Per-party update.} Party $k$'s coefficient matrices $\Gamma^{(k)}_{s_k,a_k}$ are assembled by contracting the weighted utility tensor and the quantum state against the other parties' fixed measurements. Building the operator involving the other parties' measurements involves $(\Nin\Nout)^{\Nparty-1}$ terms. Now, because the state is pure, $\rho=\ketbra{\psi}$, each term's product $\rho\hspace{0.5pt}O$ is formed as $\ket{\psi}\!\big(\bra{\psi}O\big)$---a rank-one update with only $O(d^{2\Nparty})$ operations rather than a dense $ d^{\Nparty}\times d^{\Nparty}$ matrix multiplication---before the partial trace that reduces it to $d\times d$, which requires $O(d^{\Nparty} d )=O(d^{\Nparty+1})$ operations. We repeat this for every input and output of party $k$. Hence assembling the coefficient matrices costs $O\!\big((\Nin\Nout)^{\Nparty}d^{2\Nparty}\big)$. The $\Nin\Nout$ resulting matrices are then solved by the semidefinite program~\Cref{eq:seesaw_povm_sdp} or by the positive-part projection described in~\Cref{app:seesaw-eng} for binary outputs.

\paragraph{Cost of the SDP solve.} The POVM SDP has $\Nin\Nout$ positive semidefinite blocks of size $d$. From the completeness constraints $\sum_{a_k}M_{a_k\vert s_k}=\id_d$, the SDP has exactly $d^2$ scalar equalities per input, $\Nin d^2$ in all. A primal--dual interior-point method takes Newton steps, each forming and factoring the Schur complement of these equalities. Constraints of different inputs act on disjoint blocks, so the Schur complement is block diagonal with $\Nin$ blocks of size $d^2$. Assembling these Schur complements over the $\Nout$ blocks costs $O(\Nin\Nout\,d^5)$, and factoring these blocks costs $O\!\big(\Nin(d^2)^3\big)=O(\Nin d^6)$. The PSD block factorizations cost $O(\Nin\Nout\,d^3)$. Thus the total cost is $O\!\big(\Nin d^{5}(\Nout+d)\big)$. The number of Newton steps depends on the barrier accuracy but not on $\Nparty$, so the per-party solve costs
\begin{equation*}
    C_{\mathrm{SDP}}=\tilde O\!\big(\Nin d^{5}(\Nout+d)\big),
\end{equation*}
where $\tilde O$ hides the number of Newton steps.
For binary outputs the SDP is replaced by an eigendecomposition of the $\Nin$ matrices $\Gamma^{(k)}_{s_k}$, which costs $O(\Nin d^3)$.

\paragraph{Per-sweep cost.} Summing the $\Nparty$ per-party assemblies and solves with the state assembly and update, we obtain the total cost
\begin{equation}
\begin{aligned}
    \text{general:}\quad & \tilde O\!\big(\Nparty(\Nin\Nout)^{\Nparty}d^{2\Nparty}+\Nparty\Nin d^{5}(\Nout+d)+d^{3\Nparty}\big),\label{eq:seesaw_complexity}\\
    \text{binary:}\quad & O\!\big(\Nparty(\Nin\Nout)^{\Nparty}d^{2\Nparty}+d^{3\Nparty}\big).
\end{aligned}
\end{equation}
The tilde in the general bound hides the accuracy-dependent number of interior-point Newton steps; the binary path contains no SDP solve and retains the ordinary $O$ bound.

\section{NPA Implementation Details}
\label[appendix]{app:npa-eng}
This appendix details the NPA relaxation in \Cref{eq:npa_sdp}, including the moment-matrix construction and augmented monomial set. We also analyze the moment-matrix size and give the complexity of the augmented level 1 relaxation.

\subsection{Description of algorithm}
Assume a game shape with $\Nparty$ parties, $\Nin$ inputs,
and $\Nout$ outputs. For each party~$i$, let $\mathcal O_i$
contain the $\Nin(\Nout-1)$ retained measurement projectors.
For each input, the omitted outcome projector is expressed
as the identity minus the sum of the retained projectors.

\paragraph{Moment matrices.}
An operator word is a product of measurement projectors;
its length is the number of factors.
Let $\mathcal S_k$ contain all words of length at most~$k$
formed from $\bigcup_i\mathcal O_i$, including the identity
as the empty word.

Given a quantum strategy and a chosen word set,
define the moment matrix by
\begin{equation*}
    \Gamma_{S,T}
    = \bra{\psi}S^\dagger T\ket{\psi}.
\end{equation*}
For any complex coefficients $\{c_S\}$,
\begin{equation*}
    \sum_{S,T}\overline{c_S}\,\Gamma_{S,T}\,c_T
    = \left\|\sum_S c_S S\ket{\psi}\right\|^2
    \geq 0.
\end{equation*}
Thus $\Gamma$ is positive semidefinite.
The normalization of the state also gives
$\Gamma_{\id,\id}=1$.
Further relations between the entries follow from the
projector identities and commutativity between different
parties, as specified below.
The standard NPA level~$k$ indexes this matrix by
$\mathcal S_k$; for multipartite objectives, we use the
augmented word set described next.

\paragraph{Augmented first-level monomials.}
To represent every term in a multipartite Bell objective
already at the first level, we augment the standard word set.
We fix the balanced partition
$P_1=\{1,\ldots,\lceil\Nparty/2\rceil\}$ and
$P_2=[\Nparty]\setminus P_1$.
Define the corresponding monomial sets
\begin{equation}
    \mathcal H(P_i)
    \coloneqq
    \left\{\,\prod_{j\in P_i}X_j : X_j\in\{\id\}\cup\mathcal O_j\,\right\},
    \qquad
    \mathcal S_k^{\mathrm{aug}}
    \coloneqq
    \mathcal S_k\cup\mathcal H(P_1)\cup\mathcal H(P_2).
    \label{eq:npa_aug_monomials}
\end{equation}
Expressing the last projector in terms of other projectors via completeness, every operator product that appears in the Bell operator contains either the identity or one projector for each party. It therefore factors as $u^\dagger v$, with $u\in\mathcal H(P_1)$ and $v\in\mathcal H(P_2)$, and corresponds to the $(u,v)$ element of the moment matrix $\Gamma$. In particular, $\deg(u)\le\lceil\Nparty/2\rceil$ and $\deg(v)\le\lfloor\Nparty/2\rfloor$; equality is not required because either word may contain identities. For $\Nparty=2$, both sets contain only single-party monomials already present in $\mathcal S_1$, so
$\mathcal S_1^{\mathrm{aug}}=\mathcal S_1$.

The augmentation depends only on the game shape and on the fixed balanced partition, not on the utility function. The choice of partition can affect the relaxation for multipartite games without full permutation symmetry, including the random CGLMP network and uniformly random games. Fixing the party labels makes this choice explicit and reproducible. It is immaterial for the symmetric multipartite families (MABK and Mermin) and for bipartite games. Finally, every word in $\mathcal H(P_1)\cup\mathcal H(P_2)$ has degree at most $\lceil\Nparty/2\rceil$. Hence, when $k\ge\lceil\Nparty/2\rceil$, these words already belong to the standard set $\mathcal S_k$, so $\mathcal S_k^{\mathrm{aug}}=\mathcal S_k$ and the augmentation no longer changes the relaxation.

\paragraph{Constraints.}
At level~$k$, the SDP variable is a matrix $\Gamma^{(k)}$
indexed by $\mathcal S_k^{\mathrm{aug}}$.
We require positive semidefiniteness, normalization,
and the linear relations induced by the measurement
operators. Specifically, the projectors satisfy
\begin{equation*}
\begin{aligned}
    \Pi_{i,a}(s)^\dagger
    &= \Pi_{i,a}(s), \\
    \Pi_{i,a}(s)\Pi_{i,b}(s)
    &= \delta_{a,b}\Pi_{i,a}(s), \\
    \Pi_{i,a}(s)\Pi_{j,b}(t)
    &= \Pi_{j,b}(t)\Pi_{i,a}(s),
    \qquad i\neq j.
\end{aligned}
\end{equation*}
The second identity expresses both idempotence and
orthogonality of distinct outcomes of the same input.
Completeness is incorporated by expressing each omitted
outcome projector as the identity minus the sum of the
retained projectors.

We implement these identities by reducing the products
$S^\dagger T$ that label the matrix entries.
Entries whose products reduce to the same word share
one scalar variable, and products annihilated by
orthogonality give zero entries.
For example, for a retained projector~$P$,
the identity $P^2=P$ gives
$\Gamma_{\id,P}=\Gamma_{P,P}$.

The implementation uses a real symmetric moment matrix.
Indeed, averaging a feasible matrix with its entrywise
complex conjugate preserves positive semidefiniteness,
the linear constraints, and the objective value of a
Bell functional with real coefficients.
Restricting to real matrices therefore does not change
the relaxation optimum.
Adjoint-related words can consequently share one real
scalar variable.

With these constraints, the expected utility can be
optimized by a semidefinite program.
After expanding the omitted outcomes, every term in the
expected utility is represented by an entry of
$\Gamma^{(k)}$, as shown by the factorization above.
Let $C_k$ be a real symmetric matrix encoding the
coefficients $\pi(\mathbf s)\mathcal V(\mathbf a\vert\mathbf s)$
in this representation.
The augmented level-$k$ relaxation is
\begin{equation}
\begin{aligned}
    I_k
    ={}& \max_{\Gamma^{(k)}}
         \tr\bigl(C_k\Gamma^{(k)}\bigr) \\
    \text{s.t.}\quad
    & \Gamma^{(k)} \succeq 0, \\
    & \Gamma^{(k)}_{\id,\id}=1, \\
    & \Gamma^{(k)}
      \text{ satisfies the linear relations described above}.
\end{aligned}
\label{eq:npa_sdp}
\end{equation}
Here we optimize over the matrix entries directly,
without requiring a state and measurement operators
that generate them.
Every quantum strategy gives a feasible matrix with the
same expected utility, whereas a feasible matrix need
not correspond to a quantum strategy.
Hence $I_k\geq\omega_q$.

The word sets are nested:
$\mathcal S_k^{\mathrm{aug}}
 \subseteq\mathcal S_{k+1}^{\mathrm{aug}}$.
Restricting a feasible level-$(k+1)$ matrix to the
level-$k$ index set preserves all level-$k$ constraints
and the objective value.
Therefore $I_{k+1}\leq I_k$.
Since $\mathcal S_k^{\mathrm{aug}}=\mathcal S_k$
for $k\geq\lceil\Nparty/2\rceil$, the augmented hierarchy
has the same limit as standard NPA.
This limit is the commuting-operator value, obtained
by allowing measurements of different parties to act
as commuting operators on a common Hilbert space.
It can exceed the finite-dimensional tensor-product
value~$\omega_q$~\cite{navascues2008convergent,ji2020mip}.

At low levels, the constraints in \Cref{eq:npa_sdp}
need not imply nonnegativity of every joint probability.
More precisely, the linear expressions that represent
$p(\mathbf a\vert\mathbf s)$ for genuine quantum strategies
may take negative values on a feasible relaxed matrix.
Requiring all these expressions to be nonnegative,
including those involving omitted outcomes, gives
additional linear inequalities and can strengthen the
bound~\cite{navascues2008convergent}.
Without these inequalities, a low-level bound may exceed
the algebraic maximum of the Bell expression.

\paragraph{Certification and extraction.}
At any finite NPA level, a feasible solution of the dual SDP provides an upper bound on the relaxation optimum, and therefore on \(\omega_q\); this statement is exact in exact arithmetic, while our reported numerical bounds hold to solver precision. In general, an optimal moment matrix of the relaxation need not arise from an actual quantum state and measurements. 
When an applicable rank-loop condition is satisfied,
one can reconstruct an explicit state and measurement
operators attaining the relaxation
value~\cite{navascues2008convergent}.

\subsection{Complexity analysis}
\label[appendix]{app:bigo-npa}
The standard monomial set $\mathcal S_k$ consists of products of at most $k$ measurement operators, of size $\abs{\mathcal S_k}=O([\Nparty\Nin(\Nout-1)]^k)$. A dense moment matrix requires $O(\bigl\lvert\mathcal S_k\bigr\rvert^2)$ storage, and a dense factorization within an interior-point iteration costs $O(\bigl\lvert\mathcal S_k\bigr\rvert^3)$. The complete SDP cost also depends on the number and sparsity of the affine constraints and on the number of iterations, so it is not determined by the moment-matrix dimension alone.

Our augmented first-level relaxation in~\Cref{eq:npa_aug_monomials} draws monomials from two disjoint party sets $P_1,P_2$ of sizes $\lceil\Nparty/2\rceil$ and $\lfloor\Nparty/2\rfloor$. Each party has $\kappa \coloneqq \Nin(\Nout-1)$ measurement operators, so the two monomial sets have $(1+\kappa)^{\lceil \Nparty/2\rceil}$ and $(1+\kappa)^{\lfloor \Nparty/2\rfloor}$ elements. Their only common element is the identity, so the moment-matrix dimension at this first relaxation is
\begin{equation*}
    d_\Gamma
    =(1+\kappa)^{\lceil\Nparty/2\rceil}
     +(1+\kappa)^{\lfloor\Nparty/2\rfloor}-1
    =O\!\big((1+\kappa)^{\lceil\Nparty/2\rceil}\big).
\end{equation*}
Compared with including all products containing at most one
retained projector per party, which gives
$(1+\kappa)^{\Nparty}$ indices, the balanced-partition
construction reduces the exponent to $\lceil\Nparty/2\rceil$.
The dense moment matrix therefore uses $O(d_\Gamma^2)$
storage, while each dense matrix factorization costs
\begin{equation*}
    O\!\big(d_\Gamma^{3}\big)=O\!\big((1+\Nin(\Nout-1))^{3\lceil \Nparty/2\rceil}\big).
\end{equation*}
\section{BGD Implementation Details}
\label[appendix]{app:bgd-eng}
This appendix details the components of the BGD optimizer: the batched restarts, the contraction order, the Adam-based preconditioner that sets the step scale, and the line search that adapts it across iterations.

Let $\Nparty$ be the number of parties, and assume the symmetric setting in which every party has
$\Nin$ inputs, $\Nout$ outputs, and local Hilbert-space dimension $d$.
The joint Hilbert-space dimension is
\begin{equation*}
D = d^{\Nparty}.
\end{equation*}
The complexity analysis treats $\Nout$ and $d$ as independent size parameters. In the benchmark experiments, we set $\Nout=d$ for every family except MUB, where the local dimension and Alice's output cardinality differ.

\subsection{Batched parallelism on the GPU}
\label[appendix]{app:batching}
BGD executes $q$ independent random restarts in parallel. Each restart has its own shared state and local measurement parameters. Since a single shared state has $d^{\Nparty}$ complex amplitudes, stacking the $q$ states gives a tensor of shape $q \times d^{\Nparty}$. Likewise, each local unitary is stored with an additional leading batch dimension of size $q$, so that the parameters of all $q$ strategies can be processed simultaneously. A single batched tensor contraction then evaluates the objective for all $q$ strategies and returns $q$ Bell values. On a GPU, the dominant operations are batched matrix multiplications, which can execute many restarts concurrently. As a result, increasing $q$ adds relatively little wall-clock time until the GPU is saturated; beyond that point, runtime grows roughly in proportion to the batch size. Thus, large numbers of random restarts can be used at low marginal cost on the GPU.
\Cref{fig:batch-scaling} reports the per-iteration wall-clock as a function of $q$ for a representative $(\Nparty,d,\Nin,\Nout)=(4,4,2,4)$ instance. On the A100, the time remains near $1\,\mathrm{ms}$ through $q\approx10^3$ and then grows approximately linearly as the device saturates. On one CPU core, it grows with $q$ from the start, as observed over the measured range through $q=64$.

\begin{figure}[htbp!]
\centering
\includegraphics[width=0.62\linewidth]{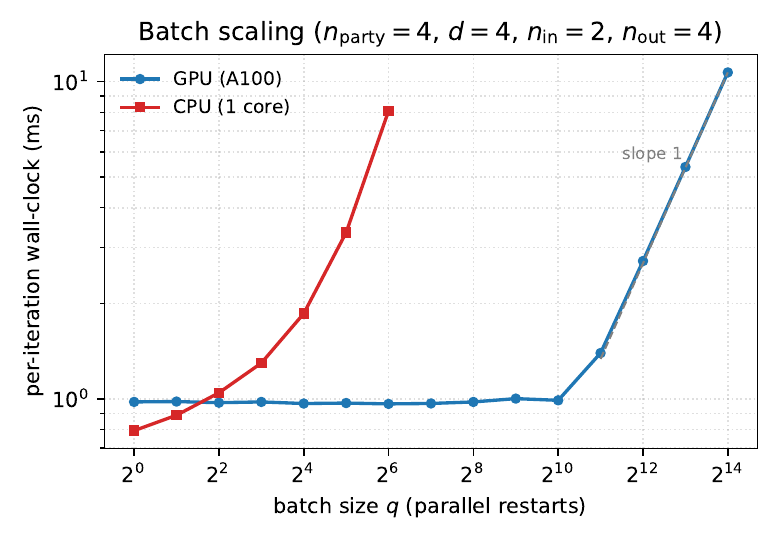}
\caption{Per-iteration wall-clock versus batch size $q$ for $(\Nparty,d,\Nin,\Nout)=(4,4,2,4)$ on log--log axes. A100 GPU and one-core CPU measurements are shown with a dashed slope-one guide.}
\label{fig:batch-scaling}
\end{figure}

\subsection{Contraction order and per-iteration cost}
\label[appendix]{app:cost}
Writing the expected utility through the Bell operator, $J=\bra{\psi}\mathcal B\ket{\psi}$ with
\begin{equation*}
\mathcal B
=
\sum_{\mathbf s,\mathbf a}
W_{\mathbf s, \mathbf a}
\bigotimes_{i=1}^{\Nparty}
\Pi^{(i)}_{a_i|s_i},
\end{equation*}
where $W_{\mathbf s,\mathbf a}=\pi(\mathbf s)V(\mathbf a|\mathbf s)$, implies a contraction order that first forms $\mathcal B$ by folding the local projectors into $W$ one party at a time and then evaluates the expectation in the state. Folding a party contracts its input-output index, of size $\Nin\Nout$, and introduces its pair of local operator indices, of size $d^2$. After $k$ folds the intermediate tensor has size
\begin{equation*}
(\Nin\Nout)^{\Nparty-k}d^{\,2k},
\end{equation*}
and folding the next party costs
\begin{equation*}
(\Nin\Nout)^{\Nparty-k}d^{\,2(k+1)}.
\end{equation*}
Summing over all $\Nparty$ folds gives
\begin{equation}
    C_{\mathcal B}
    =
    O\!\left(\sum_{j=1}^{\Nparty}
    (\Nin\Nout)^{\Nparty+1-j}d^{\,2j}\right)
    =
    O\!\left(
    \frac{\Nin\Nout\,d^2\bigl((\Nin\Nout)^{\Nparty}-d^{2\Nparty}\bigr)}
         {\Nin\Nout-d^2}
    \right),
    \label{eq:bell-operator-cost}
\end{equation}
The rational form has a removable singularity at $\Nin\Nout=d^2$, where its limit is $O\!\left(\Nparty(\Nin\Nout)^{\Nparty}d^2\right)$. The final expectation adds only $O(d^{2\Nparty})$.

BGD avoids constructing $\mathcal B$. Instead, it evaluates~\Cref{eq:bell-einsum} directly by contracting the
state tensor, the local measurement tensors, and the weighted-utility tensor $W$. For the
complexity estimate, we use a forced meet-in-the-middle contraction path that folds the local
measurement basis change for one party into the state at a time. After each fold, one local
state index of dimension $d$ is replaced by the corresponding local input-output index of size
$\Nin\Nout$.

For the direct projective path, suppose that $k$ parties have already been folded into the
state. The intermediate tensor then has size proportional to
\begin{equation*}
(\Nin\Nout)^k d^{\,\Nparty-k}.
\end{equation*}
Folding the next party contracts one $d$-dimensional state leg with a local measurement tensor
and costs
\begin{equation*}
(\Nin\Nout)^{k+1}d^{\,\Nparty-k}.
\end{equation*}
Summing over all $\Nparty$ folds gives
\begin{equation}
    C_{\mathrm{proj}}
    =
    O\!\left(\sum_{j=1}^{\Nparty}
    (\Nin\Nout)^j d^{\,\Nparty+1-j}\right)
    =
    O\!\left(
    \frac{\Nin\Nout\,d\bigl((\Nin\Nout)^{\Nparty}-d^{\Nparty}\bigr)}
         {\Nin\Nout-d}
    \right).
    \label{eq:bgd-cost-projective}
\end{equation}
Its removable singularity lies at $\Nin\Nout=d$, where the limit is
$O\!\left(\Nparty(\Nin\Nout)^{\Nparty}d\right)$. In deriving \Cref{eq:bgd-cost-projective}, $\Nout$ and $d$ are retained as independent parameters. Identifying the local basis outcomes directly with the reported outcomes then specializes the direct-read-out parameterization to $\Nout=d$; under this specialization, the singular case occurs only at $\Nin=1$.

The stochastic read-out of~\Cref{eq:stochastic} is used whenever an explicit post-processing matrix is present, and is necessary when $\Nout\neq d$. Each local
unitary is then a full $d\times d$ basis change, followed by a stochastic matrix
\begin{equation*}
S : [d] \to [\Nout]
\end{equation*}
that maps the $d$ basis outcomes to the $\Nout$ reported outputs. Along the
corresponding carry path, the factor $\Nin^{\Nparty}$ is retained throughout,
while the intermediate tensors successively trade factors of $\Nout$ for factors of
$d$. The resulting geometric series is
\begin{equation}
    C_{\mathrm{stoch}}
    =
    O\!\left(
    \Nin^{\Nparty}\sum_{k=0}^{\Nparty}
    \Nout^{\,\Nparty-k}d^{\,k+1}
    \right)
    =
    O\!\left(
    \frac{\Nin^{\Nparty}d
    \bigl(d^{\,\Nparty+1}-\Nout^{\,\Nparty+1}\bigr)}
         {d-\Nout}
    \right).
    \label{eq:bgd-cost-stochastic}
\end{equation}
Although $\Nout\neq d$ forces this read-out, an explicit $S$ may also be retained when $\Nout=d$. The corresponding limit of \Cref{eq:bgd-cost-stochastic} is therefore relevant and equals
$O\!\left((\Nparty+1)\Nin^{\Nparty}d^{\,\Nparty+1}\right)$.
The stochastic map changes the contraction structure but introduces no additional asymptotic
size parameter.

\begin{figure}[h!]
\centering
\includegraphics[width=0.85\linewidth]{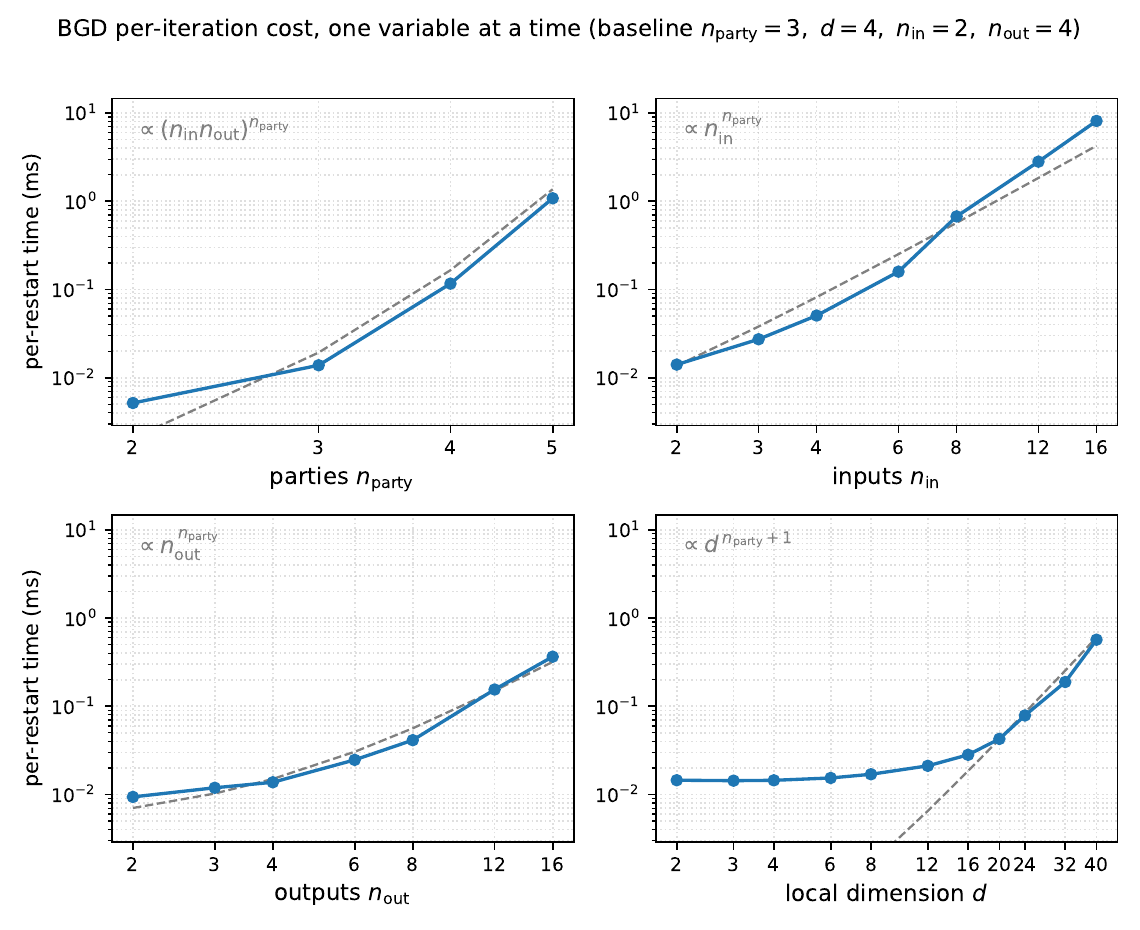}
\caption{Measured BGD per-iteration time when one size variable is varied about $(\Nparty,d,\Nin,\Nout)=(3,4,2,4)$, with both axes logarithmic. Blue curves show amortized per-restart measurements; gray dashed curves show the analytic trends described in the text.}
\label{fig:cost-scaling}
\end{figure}

\Cref{eq:bgd-cost-projective} gives the asymptotic cost of the direct-projective contraction order. For the stochastic read-out, \Cref{eq:bgd-cost-stochastic} is an upper bound obtained from the carry order; choosing a less expensive contraction order can only reduce this cost.

The Bell-operator cost~\Cref{eq:bell-operator-cost} has the same form as the direct-projective cost~\Cref{eq:bgd-cost-projective} with $d$ replaced by $d^2$: each party enters as a $d\times d$ projector rather than as a single $d$-dimensional state leg. BGD never materializes a $d^{\Nparty}\times d^{\Nparty}$ block and instead exploits the rank-one product structure of the projective-measurement terms. This distinction underlies the improved scaling with quantum dimension.

\Cref{fig:cost-scaling} compares these predictions with measured per-iteration times by varying one size variable at a time about the baseline $(\Nparty,d,\Nin,\Nout)=(3,4,2,4)$. 
We use random utility functions, input distributions, measurement bases, and shared quantum states. We then measure the average per-iteration time for our BGD optimizer.
The measurements use one pinned AMD EPYC 7543 core and a batch of $q=100$ restarts; the plotted time is the batch wall-clock divided by $q$. Both axes are logarithmic in every panel. For the four sweeps, the dashed curves evaluate the corresponding closed forms above at the baseline values of the other variables, each with its own fitted prefactor. For the dimension sweep the closed form is the stochastic upper bound~\Cref{eq:bgd-cost-stochastic}, with its prefactor fitted on the compute-bound tail $d\geq16$. Below $d\approx12$, the contraction is too small to dominate fixed dispatch overhead, and the measured time is nearly flat.

\FloatBarrier
\subsection{Preconditioning}
\label[appendix]{app:preconditioning}
BGD uses Adam with base learning rate $\eta$, supplementing its coordinate-wise preconditioning by a scalar factor for each restart. Let $v_{t,b,j}$ be the second-moment estimate for coordinate $j$ of restart $b$ at iteration $t$, and define
\begin{equation}
    c_{t,b}
    \coloneqq
    \left(
    \sum_j
    \left(\sqrt{v_{t,b,j}}+\epsilon\right)^2
    \right)^{1/2}.
    \label{eq:adam-preconditioner-norm}
\end{equation}
Let $\mathcal I_t$ denote the iterations since initialization or the most recent restoration of the optimizer state. The reference scale
\begin{equation*}
    c_{t,b}^{\max}
    =
    \max_{t'\in\mathcal I_t,\;t'\leq t}c_{t',b}
\end{equation*}
is reset whenever the parameters and Adam moments are restored. Given a threshold fraction $\kappa$, the scalar conditioning ratio is
\begin{equation}
    r_{t,b}
    =
    \frac{\min\!\left(\kappa c_{t,b}^{\max},\,c_{t,b}\right)}
         {\kappa c_{t,b}^{\max}}
    \in(0,1].
    \label{eq:effective-lr}
\end{equation}
Thus $r_{t,b}=1$ during the main descent whenever $c_{t,b}\geq\kappa c_{t,b}^{\max}$, while smaller second-moment scales give $r_{t,b}=c_{t,b}/(\kappa c_{t,b}^{\max})<1$ near convergence. With the step-size factor $\lambda_t$ of \Cref{app:linesearch}, the update applied to coordinate $j$ has scale
\begin{equation*}
    -\eta\lambda_t r_{t,b}
    \frac{\widehat m_{t,b,j}}{\sqrt{v_{t,b,j}}+\epsilon},
\end{equation*}
where $\widehat m_{t,b,j}$ is Adam's bias-corrected first-moment estimate. The restart index is retained throughout this construction, so different members of a batch can enter the damped regime at different times. When stochastic post-processing matrices are present, their parameters are scaled by the corresponding restart's $r_{t,b}$. Their gradients are first projected onto the simplex tangent cone, and every row is projected back onto the probability simplex after the update.

\subsection{Step-size calibration}
\label[appendix]{app:linesearch}
In addition to the conditioning above, BGD uses a multiplicative factor $\lambda_t$, initialized at $\lambda_0=1$ and decreased monotonically. The initial trainable parameters and Adam moments are retained before optimization begins; no later iterate replaces this reference state.

This mechanism calibrates the step size from the initial batch of strategies. If 12 cumulative improvements occur while the initial reference state is retained, then
\begin{equation*}
    \lambda_t\leftarrow \tfrac12\lambda_t,
\end{equation*}
the parameters and Adam moments are restored to their initial values, and the reference state is discarded. Each restart then resumes from its own initial strategy at the reduced scale. Because the improvement rule requires the retained reference state, this halving can occur at most once. If instead a window of 12 iterations contains no improvement, then
\begin{equation*}
    \lambda_t\leftarrow \tfrac23\lambda_t.
\end{equation*}
While the initial reference state is retained, this event also restores the initial parameters and Adam moments; after that state has been discarded, subsequent no-improvement windows reduce $\lambda_t$ without restoring the optimizer state. The restart-specific reference scales in \Cref{eq:effective-lr} are reset whenever such a restoration occurs.

Restoring the Adam moments together with the parameters prevents information accumulated along a rejected trajectory from influencing the renewed descent. BGD also retains the best feasible strategy encountered over all iterations and all batch elements, independently of the current iterate used by the calibration procedure.

\section{Representations of Bell Inequalities}
\label[appendix]{app:expressions}
The evaluated families are written in several standard forms: full joint probabilities, marginal probabilities, and correlators. This appendix records how these forms relate to the expected-utility convention of \Cref{section:nolocalgame}.

For input and output tuples $\mathbf s\in\prod_i S_i$ and $\mathbf a\in\prod_i A_i$, respectively, that convention is
\begin{equation*}
    \bar{\mathcal V}
    =\sum_{\mathbf s}\sum_{\mathbf a}
    \pi(\mathbf s)\,\mathcal V(\mathbf a\vert\mathbf s)\,
    p(\mathbf a\vert\mathbf s)
    \leq \omega_c.
\end{equation*}
Thus the Bell coefficients are the entries of the weighted utility tensor $\pi(\mathbf s)\mathcal V(\mathbf a\vert\mathbf s)$.
As a linear functional of the behavior, a Bell inequality defines a hyperplane separating the local set from behaviors that violate its classical bound.

Two properties allow conversion between the forms below. First, every behavior is normalized:
\begin{equation*}
    \sum_{\mathbf a\in\prod_i A_i}p(\mathbf a\vert\mathbf s)=1
    \qquad
    \forall\,\mathbf s\in\prod_i S_i.
\end{equation*}
Second, a quantum behavior is no-signaling. For every party $i$, fixed $\mathbf a_{-i}$ and $\mathbf s_{-i}$, and inputs $s_i,s_i'\in S_i$,
\begin{equation}
    \sum_{a_i\in A_i}p(a_i,\mathbf a_{-i}\vert s_i,\mathbf s_{-i})
    =
    \sum_{a_i\in A_i}p(a_i,\mathbf a_{-i}\vert s_i',\mathbf s_{-i}).
    \label{eq:nosignaling_constraints}
\end{equation}
Consequently, for every nonempty proper subset $I\subsetneq[\Nparty]$, the marginal behavior is well defined by
\begin{equation}
    p(\mathbf a_I\vert\mathbf s_I)
    =
    \sum_{\mathbf a_{I^c}\in\prod_{i\in I^c}A_i}
    p(\mathbf a_I,\mathbf a_{I^c}\vert
      \mathbf s_I,\mathbf s_{I^c}),
    \label{eq:marginal}
\end{equation}
and the right-hand side is independent of the choice of $\mathbf s_{I^c}$.

\paragraph{Full joint-probability form.}
Every Bell functional can be written as
\begin{equation*}
    \mathcal B[p]
    =\sum_{\mathbf s}\sum_{\mathbf a}
    \beta_{\mathbf s,\mathbf a}\,
    p(\mathbf a\vert\mathbf s)
    \leq\omega_c,
    \qquad \beta_{\mathbf s,\mathbf a}\in\R.
\end{equation*}
This is exactly the expected-utility form after identifying $\beta_{\mathbf s,\mathbf a}=\pi(\mathbf s)\mathcal V(\mathbf a\vert\mathbf s)$.

\paragraph{Marginal-probability form.}
A Bell functional may also contain terms $p(\mathbf a_I\vert\mathbf s_I)$ involving only a subset of the parties. Substituting \Cref{eq:marginal} converts each such term into full joint probabilities. No-signaling guarantees that the result does not depend on the unspecified inputs of the complementary parties.

\paragraph{Correlator form.}
When every party has $\Nout$ outputs, let $\zeta=e^{2\pi\mathrm{i}/\Nout}$. For two parties, the correlators are
\begin{equation*}
    \begin{aligned}
        \langle A_x^k\rangle
        &\coloneqq\sum_{a=0}^{\Nout-1}\zeta^{ak}p_A(a\vert x),\\
        \langle B_y^l\rangle
        &\coloneqq\sum_{b=0}^{\Nout-1}\zeta^{bl}p_B(b\vert y),\\
        \langle A_x^kB_y^l\rangle
        &\coloneqq\sum_{a,b=0}^{\Nout-1}\zeta^{ak+bl}p(a,b\vert x,y).
    \end{aligned}
\end{equation*}
For a quantum strategy, these correlators are expectations of powers of operators whose eigenvalues are the $\Nout$-th roots of unity.
The multipartite definition is analogous. For orders $k_i\in\{0,\ldots,\Nout-1\}$,
\begin{equation*}
    \big\langle A^{(1),k_1}_{x_1}\cdots
    A^{(\Nparty),k_{\Nparty}}_{x_{\Nparty}}\big\rangle
    =\sum_{a_1,\ldots,a_{\Nparty}=0}^{\Nout-1}
    \zeta^{\sum_i k_i a_i}
    p(a_1,\ldots,a_{\Nparty}\vert x_1,\ldots,x_{\Nparty}),
\end{equation*}
while correlators over fewer parties use the corresponding marginal behavior.

For binary outputs, the CHSH inequality is
\begin{equation}
    \langle A_0B_0\rangle+\langle A_0B_1\rangle
    +\langle A_1B_0\rangle-\langle A_1B_1\rangle\leq2.
    \label{eq:chsh_corr}
\end{equation}
Since
\begin{equation*}
    \langle A_xB_y\rangle
    =2p(a=b\vert x,y)-1
    =1-2p(a\neq b\vert x,y),
\end{equation*}
\Cref{eq:chsh_corr} is equivalent to the CHSH-game form
\begin{equation*}
    \frac14\bigl[
    p(a=b\vert0,0)+p(a=b\vert 0,1)+p(a=b\vert 1,0)
    +p(a\neq b\vert 1,1)\bigr]
    \leq\frac34.
\end{equation*}

\section{Proof of Two-Input Load Balancing Result}
\label[appendix]{app:multioutput-chsh}

This appendix proves \Cref{thm:multi-output-chsh}.
Let \(m:=\Nout\in\N\), with \(m\ge2\), denote the common output-alphabet size.
Alice and Bob receive uniformly random inputs \(x,y\in\{0,1\}\) and return outputs
\begin{equation*}
    a,b\in\{0,1,\ldots,m-1\}.
\end{equation*}
They win if \(a=b\) for the three input pairs \((x,y)\neq(1,1)\), and if \(a\neq b\) for \((x,y)=(1,1)\).  This is the equality--disequality version of CHSH, with an arbitrary output alphabet.

For any behavior \(p(a,b\vert x,y)\), set
\begin{equation}
    p^{=}_{xy}:=\Pr[a=b\vert x,y]
    =\sum_{r=0}^{m-1}p(r,r\vert x,y),
    \qquad
    Q_m:=p^{=}_{00}+p^{=}_{01}+p^{=}_{10}-p^{=}_{11}.
    \label{eq:lbchsh-Q-def}
\end{equation}
The winning probability is
\begin{equation}
    \omega=\frac14\bigl(p^{=}_{00}+p^{=}_{01}+p^{=}_{10}+1-p^{=}_{11}\bigr)
    =\frac{1+Q_m}{4}.
    \label{eq:lbchsh-winning}
\end{equation}
We first prove the classical part of \Cref{thm:multi-output-chsh}. By \Cref{eq:lbchsh-winning}, its quantum part is equivalent to establishing \(Q_m\le1+\sqrt2\) for all finite-dimensional tensor-product strategies and to exhibiting a strategy attaining equality. The binary CHSH strategy provides the lower bound. The remaining task is to exclude any improvement from additional output labels. We do this through an operator inequality for the four projectors associated with one label. Its right-hand side is linear in the local projectors, so summing over labels removes the dependence on \(m\) by measurement completeness.

\subsection{Classical value}

A deterministic classical strategy fixes four labels \(a_0,a_1,b_0,b_1\).  If it wins the first three input pairs, then
\begin{equation*}
    a_0=b_0,\qquad a_0=b_1,\qquad a_1=b_0.
\end{equation*}
By transitivity, \(a_1=b_1\), which loses on input \((1,1)\).  Thus no deterministic strategy wins all four inputs.  Since classical strategies are convex combinations of deterministic ones, \(\omega_c\le3/4\).  The constant strategy \(a_x=b_y=0\) wins the first three inputs and loses only \((1,1)\), so \(\omega_c=3/4\).

\subsection{Quantum model and projective reduction}

Alice and Bob share a state \(\rho\) on \(\mathcal H_A\otimes\mathcal H_B\).  For each input \(x\), Alice has a POVM \(\{E_x^r\}_{r=0}^{m-1}\), and for each input \(y\), Bob has a POVM \(\{F_y^r\}_{r=0}^{m-1}\).  The probabilities are
\begin{equation*}
    p(r,s\vert x,y)=\operatorname{Tr}\bigl[\rho(E_x^r\otimes F_y^s)\bigr].
\end{equation*}
Since only finitely many POVMs occur, a common local Naimark dilation lets us assume, without changing any value, that all these measurements are projective.  For one party, dilate each POVM input \(x\) on \(\mathcal H\) to a PVM on \(\mathcal H\oplus\mathcal K_x\).  The direct sum
\begin{equation*}
    \widehat{\mathcal H}:=\mathcal H\oplus\mathcal K_0\oplus\mathcal K_1
\end{equation*}
contains a common copy of \(\mathcal H\) and remains finite-dimensional.  Extend the PVM for input \(0\) arbitrarily on the unused \(\mathcal K_1\) summand, and extend the PVM for input \(1\) arbitrarily on \(\mathcal K_0\).  The resulting PVMs act on the same \(\widehat{\mathcal H}\), and their compressions to \(\mathcal H\) recover the original POVMs.  Applying this construction independently to Alice and Bob, and embedding \(\rho\) by the two local isometries, preserves every joint probability.  It is therefore enough to prove the upper bound for finite-dimensional projective measurements.

Write Alice's projectors as \(\{A_x^r\}_{r=0}^{m-1}\) and Bob's as \(\{B_y^r\}_{r=0}^{m-1}\), with
\begin{equation*}
    A_x^rA_x^s=\delta_{rs}A_x^r,\qquad
    \sum_{r=0}^{m-1}A_x^r=I_A,
    \qquad
    B_y^rB_y^s=\delta_{rs}B_y^r,\qquad
    \sum_{r=0}^{m-1}B_y^r=I_B.
\end{equation*}
Tensor products are suppressed below, so \(A_x^rB_y^s\) denotes \(A_x^r\otimes B_y^s\).  The Bell operator for \(Q_m\) is
\begin{equation}
    \mathcal B_m
    :=\sum_{r=0}^{m-1}
    \left(
        A_0^rB_0^r+A_0^rB_1^r+A_1^rB_0^r-A_1^rB_1^r
    \right),
    \label{eq:lbchsh-bell-operator}
\end{equation}
and \(Q_m=\operatorname{Tr}(\rho\mathcal B_m)\).

\subsection{A four-projection certificate}

The upper bound is obtained by applying the following operator inequality to each output label.

\begin{lem}[Four-projection certificate]
\label[lem]{lem:lbchsh-four-projection}
Let \(P,Q\) be projections on Alice's space and let \(R,S\) be projections on Bob's space.  In the tensor-product representation,
\begin{equation}
    PR+PS+QR-QS
    \preceq
    \frac{2+\sqrt2}{4}(P+R)+\frac{\sqrt2}{4}(Q+S).
    \label{eq:lbchsh-four-projection}
\end{equation}
\end{lem}

Assuming the lemma, the upper bound follows immediately.  For each output \(r\), substitute
\begin{equation*}
    P=A_0^r,\qquad Q=A_1^r,\qquad R=B_0^r,\qquad S=B_1^r
\end{equation*}
in \Cref{lem:lbchsh-four-projection}.  Summing over \(r\) and using PVM completeness gives
\begin{equation}
\begin{aligned}
    \mathcal B_m
    &\preceq
    \frac{2+\sqrt2}{4}\sum_{r=0}^{m-1}(A_0^r+B_0^r)
    +\frac{\sqrt2}{4}\sum_{r=0}^{m-1}(A_1^r+B_1^r) \\
    &=
    \left(2\cdot\frac{2+\sqrt2}{4}+2\cdot\frac{\sqrt2}{4}\right)I
    =(1+\sqrt2)I .
\end{aligned}
    \label{eq:lbchsh-operator-upper}
\end{equation}
Thus \(Q_m\le1+\sqrt2\) for every finite-dimensional tensor-product quantum strategy.

\subsection{Proof of the certificate}

It remains to prove \Cref{lem:lbchsh-four-projection}.  Define self-adjoint unitaries
\begin{equation*}
    X:=2P-I_A,\qquad
    Y:=2Q-I_A,\qquad
    U:=2R-I_B,\qquad
    V:=2S-I_B.
\end{equation*}
A direct expansion gives
\begin{equation}
    PR+PS+QR-QS
    =\frac12I+\frac12(X+U)+\frac14(XU+XV+YU-YV),
    \label{eq:lbchsh-left-expanded}
\end{equation}
while
\begin{equation}
    \frac{2+\sqrt2}{4}(P+R)+\frac{\sqrt2}{4}(Q+S)
    =
    \frac{1+\sqrt2}{2}I
    +\frac{2+\sqrt2}{8}(X+U)
    +\frac{\sqrt2}{8}(Y+V).
    \label{eq:lbchsh-right-expanded}
\end{equation}
Let
\begin{equation}
    H:=XU+XV+YU-YV+\frac{2-\sqrt2}{2}(X+U)-\frac{\sqrt2}{2}(Y+V).
    \label{eq:lbchsh-H-def}
\end{equation}
Subtracting \Cref{eq:lbchsh-left-expanded} from \Cref{eq:lbchsh-right-expanded} shows that the slack in \Cref{eq:lbchsh-four-projection} equals
\begin{equation*}
    \frac14(2\sqrt2 I-H).
\end{equation*}
Hence it suffices to prove
\begin{equation}
    H\preceq2\sqrt2 I.
    \label{eq:lbchsh-H-bound}
\end{equation}

The pairs \(X,Y\) and \(U,V\) are pairs of self-adjoint unitaries.  In finite dimensions, Jordan's lemma decomposes each local space into common invariant subspaces of dimension at most two.  On a two-dimensional Alice block, after choosing phases, we may write
\begin{equation}
    X=
    \begin{pmatrix}
    1&0\\
    0&-1
    \end{pmatrix},
    \qquad
    Y=
    \begin{pmatrix}
    s&\sqrt{1-s^2}\\
    \sqrt{1-s^2}&-s
    \end{pmatrix},
    \qquad
    s\in[-1,1],
    \label{eq:lbchsh-XY-block}
\end{equation}
and similarly on a two-dimensional Bob block,
\begin{equation}
    U=
    \begin{pmatrix}
    1&0\\
    0&-1
    \end{pmatrix},
    \qquad
    V=
    \begin{pmatrix}
    t&\sqrt{1-t^2}\\
    \sqrt{1-t^2}&-t
    \end{pmatrix},
    \qquad
    t\in[-1,1].
    \label{eq:lbchsh-UV-block}
\end{equation}
Each one-dimensional Alice block is a scalar sign pair obtained as a principal subblock of \Cref{eq:lbchsh-XY-block} at \(s=\pm1\).  Likewise, every one-dimensional Bob block is a principal subblock of \Cref{eq:lbchsh-UV-block} at \(t=\pm1\).  Thus the local block types are \(2\otimes2\), \(1\otimes2\), \(2\otimes1\), and \(1\otimes1\).  The corresponding joint operator matrices have sizes \(4\times4\), \(2\times2\), \(2\times2\), and \(1\times1\), respectively.  The lower-dimensional cases are principal submatrices of the \(4\times4\) matrices at endpoint values of \(s\) or \(t\).  Hence, it suffices to prove positivity of the \(4\times4\) matrices for all \(s,t\in[-1,1]\).

Put \(\eta:=\sqrt2\), \(r_s:=\sqrt{1-s^2}\), and \(r_t:=\sqrt{1-t^2}\).  On a two-dimensional Alice block and a two-dimensional Bob block, define
\begin{equation*}
    M(s,t):=2\eta I-H(s,t).
\end{equation*}
In the basis \(|00\rangle,|01\rangle,|10\rangle,|11\rangle\), the non-redundant entries of the real symmetric matrix \(M(s,t)\) are
\begin{equation}
\begin{aligned}
M_{11}&=st+\left(-1+\frac{\eta}{2}\right)(s+t)-3+3\eta,\label{eq:lbchsh-M11}\\
M_{12}&=\frac{r_t(2s-2+\eta)}{2},
&
M_{13}&=\frac{r_s(2t-2+\eta)}{2},
&
M_{14}&=r_sr_t,\\
M_{22}&=-st+s+t+\frac{\eta}{2}(s-t)+1+2\eta,\\
M_{23}&=r_sr_t,
&
M_{24}&=\frac{r_s(-2t+\eta+2)}{2},\\
M_{33}&=-st+s+t+\frac{\eta}{2}(-s+t)+1+2\eta,\\
M_{34}&=\frac{r_t(-2s+\eta+2)}{2},\\
M_{44}&=st-s-t-\frac{\eta}{2}(s+t)+1+\eta.
\end{aligned}
\end{equation}
Let \(\Delta_j\) denote the leading principal minor of order \(j\).  After using \(r_s^2=1-s^2\) and \(r_t^2=1-t^2\), the first three are
\begin{equation*}
    \Delta_1
    =st+\left(-1+\frac{\eta}{2}\right)(s+t)-3+3\eta,
\end{equation*}
\begin{equation*}
    \Delta_2
    =\frac12\left(
    4s^2t-3s^2+4st-2\eta s+6s-10t+6\eta t-4\eta+15
    \right),
\end{equation*}
\begin{align}
    \Delta_3=\frac12\Bigl(&
    8s^2t^2-4\eta s^2t+10s^2t-12s^2+3\eta s^2
    -4\eta st^2+10st^2 \nonumber\\
    &-16st+22\eta st+14s-12t^2+3\eta t^2+14t+16+12\eta
    \Bigr).
\end{align}
To check their signs on the square, set \(x=(s+1)/2\) and \(z=(t+1)/2\).  For \(w\in[0,1]\), let \(B_i^n(w):=\binom ni w^i(1-w)^{n-i}\).  These Bernstein polynomials are nonnegative and form a partition of unity.  For a matrix \(C=(c_{ij})_{0\le i\le d,\,0\le j\le e}\), we say that \(C\) is the tensor-product Bernstein coefficient matrix of a polynomial \(p\) when
\[
    p(s,t)=\sum_{i=0}^{d}\sum_{j=0}^{e}
    c_{ij}B_i^d\!\left(\frac{s+1}{2}\right)
    B_j^e\!\left(\frac{t+1}{2}\right).
\]
With this convention, the Bernstein coefficient matrices of \(\Delta_1\), \(\Delta_2\), and \(\Delta_3\) are respectively
\begin{equation}
\begin{pmatrix}
2\eta & -4+3\eta\\
-4+3\eta & 4(\eta-1)
\end{pmatrix},
\label{eq:lbchsh-bernstein-Delta1}
\end{equation}
\begin{equation}
\begin{pmatrix}
4(2-\eta) & 2(\eta-1)\\
16-5\eta & 2+\eta\\
2(5-3\eta) & 8
\end{pmatrix},
\label{eq:lbchsh-bernstein-Delta2}
\end{equation}
and
\begin{equation}
\begin{pmatrix}
8(-4+3\eta) & 2(1+2\eta) & 2(4-\eta)\\
2(1+2\eta) & 3(\eta+8) & 2(3+4\eta)\\
2(4-\eta) & 2(3+4\eta) & 16(1+\eta)
\end{pmatrix}.
\label{eq:lbchsh-bernstein-Delta3}
\end{equation}
All entries in these three matrices are strictly positive for \(\eta=\sqrt2\).
The tensor-product Bernstein basis is nonnegative and forms a partition of
unity on the square.  Therefore,
\begin{equation}
    \Delta_1>0,\qquad
    \Delta_2>0,\qquad
    \Delta_3>0
    \qquad
    \text{for all }(s,t)\in[-1,1]^2.
    \label{eq:lbchsh-first-three-minors}
\end{equation}

The determinant is
\begin{equation}
\begin{aligned}
    \Delta_4={}&
    -16s^2t^2+8s^2t+8st^2+3s^2-6st+3t^2 \nonumber\\
    &-8\eta s^2t-8\eta st^2+6\eta s^2+4\eta st+6\eta t^2.  
\end{aligned}
\label{eq:lbchsh-Delta4-st}
\end{equation}
With \(u=s+t\) and \(v=st\), this becomes
\begin{equation}
    \Delta_4
    =(3+6\eta)u^2+8(1-\eta)uv-(12+8\eta)v-16v^2.
    \label{eq:lbchsh-Delta4-uv}
\end{equation}
For \(s,t\in[-1,1]\) and fixed \(u\in[-2,2]\), every feasible value of \(v\) lies in the interval
\begin{equation}
    |u|-1\le v\le \frac{u^2}{4}.
    \label{eq:lbchsh-uv-feasible}
\end{equation}
The upper bound follows from \((s-t)^2\ge0\).  The lower bound follows from \((1-s)(1-t)\ge0\) and \((1+s)(1+t)\ge0\), which imply \(v\ge u-1\) and \(v\ge -u-1\).  Thus nonnegativity on this enclosing interval implies nonnegativity on the actual feasible set.

As a polynomial in \(v\), \Cref{eq:lbchsh-Delta4-uv} is concave, so its minimum over the enclosing interval \Cref{eq:lbchsh-uv-feasible} is attained at an endpoint.  At \(v=u^2/4\),
\begin{equation*}
    \left.\Delta_4\right|_{v=u^2/4}
    =-u^2(u-2)(u+2\eta)
    \ge0,
    \qquad -2\le u\le2.
\end{equation*}
At \(v=u-1\), for \(0\le u\le2\),
\begin{equation*}
    \left.\Delta_4\right|_{v=u-1}
    =-(5+2\eta)(u-2)
    \left(u-\frac{26-24\eta}{17}\right)
    \ge0.
\end{equation*}
At \(v=-u-1\), for \(-2\le u\le0\),
\begin{equation*}
    \left.\Delta_4\right|_{v=-u-1}
    =(14\eta-21)
    \left(u-\frac{2+4\eta}{7}\right)
    \left(u+\frac{22+12\eta}{7}\right)
    \ge0.
\end{equation*}
Here the signs use only \(\eta=\sqrt2\): \(26-24\eta<0\), \(14\eta-21<0\), \((2+4\eta)/7>0\), and \((22+12\eta)/7>2\).  Therefore \(\Delta_4\ge0\) on the whole square.

Since the first three leading minors are strictly positive, Gaussian elimination without row or column exchanges gives the no-pivot LDL factorization
\begin{equation}
    M=L\,
    \operatorname{diag}\!\left(
    \Delta_1,\,
    \frac{\Delta_2}{\Delta_1},\,
    \frac{\Delta_3}{\Delta_2},\,
    \frac{\Delta_4}{\Delta_3}
    \right)L^T,
    \label{eq:lbchsh-LDL}
\end{equation}
where \(L\) is real unit lower triangular and hence invertible.  The first three pivots are strictly positive, while the last pivot is nonnegative.  Congruence by the invertible matrix \(L\) therefore gives \(M(s,t)\succeq0\) for all \(s,t\in[-1,1]\).  This proves \(2\sqrt2 I-H\succeq0\), and therefore \Cref{lem:lbchsh-four-projection}.

\subsection{Achievability}

For \(m\ge2\), embed the usual binary CHSH strategy into the first two labels.  Alice and Bob share
\begin{equation*}
    |\Phi^+\rangle=\frac{|00\rangle+|11\rangle}{\sqrt2}.
\end{equation*}
Alice measures the Pauli \(Z\) observable for input \(0\) and the Pauli \(X\) observable for input \(1\).  Bob measures \((Z+X)/\sqrt2\) for input \(0\) and \((Z-X)/\sqrt2\) for input \(1\).  The \(+1\) eigenvalue is reported as output \(0\), and the \(-1\) eigenvalue is reported as output \(1\).  For each input, the projectors assigned to labels \(2,\ldots,m-1\) are zero.  Zero effects are allowed in a finite-outcome PVM, so this defines a valid \(m\)-output strategy.

The resulting binary correlators are
\begin{equation*}
    E_{00}=E_{01}=E_{10}=\frac1{\sqrt2},
    \qquad
    E_{11}=-\frac1{\sqrt2}.
\end{equation*}
Thus
\begin{equation*}
    p^{=}_{00}=p^{=}_{01}=p^{=}_{10}
    =\frac{1+1/\sqrt2}{2}
    =\cos^2\frac{\pi}{8},
    \qquad
    p^{=}_{11}
    =\frac{1-1/\sqrt2}{2}
    =\sin^2\frac{\pi}{8}.
\end{equation*}
Substituting into \Cref{eq:lbchsh-Q-def} gives \(Q_m=1+\sqrt2\).  Together with the upper bound \Cref{eq:lbchsh-operator-upper}, this proves \Cref{thm:multi-output-chsh}.

\section{Random CGLMP Network Instances}
\label[appendix]{app:cglmpgraph-instances}
\Cref{fig:cglmpgraph} specifies the two random CGLMP network games used in \Cref{tab:eval-unknown}. Parties are represented by nodes and CGLMP functionals by edges. The string $c_0c_1\cdots c_{\Nout-1}$ at an endpoint specifies its output relabeling $\pi(k)=c_k$ for that edge.

The left panel shows the star graph used for the $(4,2,4)$ game, in which party~2 is incident to all three CGLMP functionals. BGD attains the feasible value $q_\downarrow=0.54602$, which exceeds the exact classical value $1/2$ and therefore certifies a quantum advantage.

The right panel shows the path $1$--$2$--$3$--$0$ used for the $(4,2,6)$ game. BGD attains $q_\downarrow=0.58762>1/2$, likewise certifying a quantum advantage.

\begin{figure}[h!]
\centering
\includegraphics[width=\linewidth]{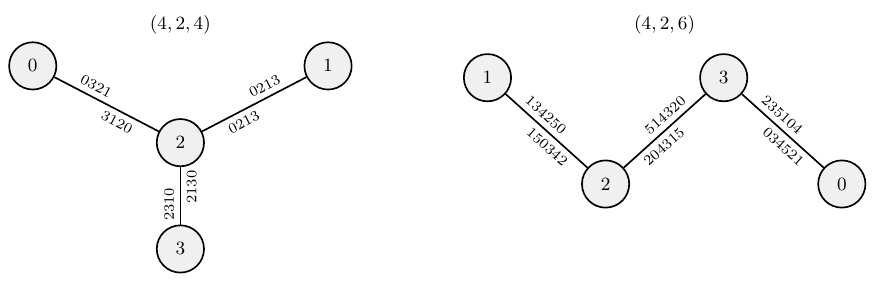}
\caption{The random CGLMP network games used in \Cref{tab:eval-unknown}, labeled by the game shape $(\Nparty,\Nin,\Nout)$: $(4,2,4)$ (left) and $(4,2,6)$ (right). The two endpoint strings on each edge give the output relabelings for its CGLMP functional; the string closer to a party specifies that party's permutation.}
\label{fig:cglmpgraph}
\end{figure}
\addtocontents{toc}{\protect\setcounter{tocdepth}{2}}
\end{document}